\documentclass[a4paper,12pt]{article}
\ifx\TeXtures\undefined
\usepackage[psamsfonts]{amsfonts}%
\usepackage{amsmath}%
\else
\usepackage{amsmath}%
\usepackage{amsfonts}%
\usepackage{times}
\fi
\usepackage{amsthm} 
\usepackage{epsfig}      
\usepackage{color}
\usepackage{psfrag}

\newcommand{\DATUM}{February-8-2010}               
\newcommand{\change}
{{\marginpar{\#}}}        
\newcommand{\comma}{\: ,}              
\newcommand{\Om}{\Omega}                

\newcommand{\one}{{\bf 1}}
\newcommand{\cA}{{\mathcal{A}}}

\newcommand{\cF}{{\mathcal{F}}}

\newcommand{\cH}{{\mathcal{H}}}
\newcommand{\cI}{{\mathcal{I}}}

\newcommand{\cO}{{\mathcal{O}}}         

\newcommand{\cR}{{\mathcal{R}}}
\newcommand{\cS}{{\mathcal{S}}}

\newcommand{\RR}{\mathbb{R}}            
\newcommand{\NN}{\mathbb{N}}            
\newcommand{\CC}{\mathbb{C}}            
\newcommand{\fh}{\mathfrak{h}}         
\newcommand{\vnabla}{{\vec{\nabla}}}

\newcommand{\vP}{{\vec{P}}}

\newcommand{\vk}{{\vec{k}}}

\newcommand{\vp}{{\vec{p}}}

\newcommand{\vx}{{\vec{x}}}

\newcommand{\vz}{{\vec{z}}}
\newcommand{\vy}{{\vec{y}}}

\newcommand{\dom}{{\rm dom}}

\newcommand{\cirS}{\mathop{\bigcirc\kern -.73em {\scriptstyle{\rm S}}}}

\newcommand{\QED}{\phantom{blablabla}\hfill\qed\newline}  
\renewcommand{\thesection}
{\Roman{section}}                      
\renewcommand{\theequation}
{\thesection.\arabic{equation}}        

\newtheorem{theorem}{Theorem}[section]         
\newtheorem{lemma}[theorem]{Lemma}             
\newtheorem{corollary}[theorem]{Corollary}     

\renewcommand{\theequation}{\thesection.\arabic{equation}}
\newcommand{\resetequ}{\setcounter{equation}{0}}

\newcommand{\betacrit}{{11/2}}  
\renewcommand{\dom}{\mathrm{Dom}}  
\newcommand{\ep}{\epsilon}  
  
\newcommand{\vdiffdilation}{\vec{\mathfrak{d}}  }  
\newcommand{\str}{|}  
\newcommand{\norm}{\|}  
\newtheorem*{blanktheorem}{Theorem}  
\newtheorem*{blanklemma}{Lemma}   
\newcommand{\trialvector}{\eta_{\vec{Q}}}

\begin{document}
\bibliographystyle{plain}
\setcounter{page}{0}
\thispagestyle{empty}

\author{ Wojciech  De Roeck \\
\small{ Institut f\"ur Theoretische Physik  Universit\"at Heidelberg, 
 }   \\[-1ex] 
\small{Philosophenweg 19  D69120 Heidelberg,  Germany } \\
(w.deroeck@thphys.uni-heidelberg.de)
\and
J\"{u}rg Fr\"{o}hlich \\      
\small{Institute~of~Theoretical Physics; ETH Z\"{u}rich;} \\[-1ex] 
\small{CH-8093 Z\"{u}rich, Switzerland} \\
(juerg@itp.phys.ethz.ch)
\and
Alessandro Pizzo\\
\small{Department of Mathematics, University of California Davis;} \\[-1ex] 
\small{One Shields Avenue, Davis, California 95616, USA} \\
(pizzo@math.ucdavis.edu)
}
\date{\DATUM}

\setcounter{page}{0}
\thispagestyle{empty}
\title{Absence of Embedded  Mass Shells: Cerenkov Radiation and Quantum Friction}

\maketitle

\begin{abstract}
We show that, in a model where a non-relativistic
particle is coupled to a quantized relativistic scalar Bose field, the embedded mass shell of the
particle dissolves in the continuum when the interaction is turned on, provided the coupling constant is sufficiently small. More precisely, under the assumption that
the fiber eigenvectors corresponding to the putative mass
shell are differentiable as functions of the total momentum of the system, we show that a
mass shell could exist only at a strictly positive distance  from the unperturbed embedded mass
shell near the boundary of the energy-momentum spectrum. 
\end{abstract}
%

\thispagestyle{empty}
\newpage
\setcounter{page}{1}
\section{Introduction}\label{sec: intro}
The model studied in this paper describes a system consisting of a non-relativistic quantum particle coupled to a quantized relativistic field of scalar massless bosons through an interaction term linear in creation- and annihilation operators. The system is invariant under space translations.  Therefore its total momentum is conserved. In states  where the initial particle momentum is larger than $mc$, where $m$ is the mass of the non-relativistic particle and $c$ the propagation speed of the bosonic modes, we expect that the particle will emit Cerenkov radiation,  because its group velocity is larger than the speed of the bosons. We are thus interested in the spectral region $(E,\vP)$ with $|\vP|>1$; using units such that $m=c=1$.  Here $E, \vP$ are the spectral variables of the Hamiltonian and of the total momentum operator, respectively. In this region, we expect that a mass shell of the non-relativistic particle does not exist. Put differently,  we expect that the mass shell,  which in the unperturbed system is described by the equation $E=\vP^2/2$,  disappears,  as soon as the interaction is switched on. This would show that one-particle states of the non-relativistic particle are unstable for values of $|\vP|$ larger than $1$. 

Our main result is as follows. We assume that, for  $|\vP|>1$, a mass shell exists  with the property that the corresponding fiber eigenvectors are differentiable as functions of the total momentum of the system. Then we show that,  for sufficiently small values of the coupling constant, such a
mass shell may exist only at a strictly positive distance ($>\cO(1)$)  from the unperturbed mass
shell in the energy-momentum spectrum. More precisely, one-particle states might only exist in a region around the three-dimensional surface $E=|\vP|-\frac{1}{2}$, whose width tends to zero,  as the coupling constant approaches $0$. Our results are proven for models with a fixed ultraviolet cutoff that turns off interactions with high-energy bosons, and under the assumption of appropriate infrared regularity of the form factor that models the interaction.

In the literature, many results are concerned with the existence of a mass shell for $|\vP|<1$, depending on the  behavior of the coupling between the non-relativistic particle and the relativistic boson field in the infrared region. These results clarify and extend the notion of stable particle by providing a scattering picture for \emph{infraparticles}, for which a mass shell does not exist (i.e., the single-particle states are not normalizable in the Hilbert space of  pure states of the system); see \cite{F1}, \cite{F2}, \cite{P1}, \cite{P2}, \cite{C}, \cite{CFP1}, \cite{CFP2}, \cite{BCFS2}, \cite{FP2}, \cite{hasler}.

To our knowledge, for the spectral region studied in this paper, no rigorous results have yet appeared in the literature concerning the existence or non-existence of an embedded mass shell. However,  in \cite{erdos}, for the model studied in this  paper,  it is proven that  the electron motion in the kinetic limit is described by a Boltzmann equation that exhibits  the slowdown of the particle by emitting Cerenkov radiation, as long as its velocity is greater than $1$.  This supports the thesis that there is no mass shell for $\str \vP \str >1$.  

We also stress that the conclusions of our paper leave open an interesting question: Our analysis does not exclude the existence of  single-particle states near the boundary of the energy-momentum spectrum (which,  for $|\vP|>1$,  is approximately linear in $|\vP|$).  In this respect, we recall that the existence of the groundstate  eigenvalue for the fiber Hamiltonians,  in the region $|\vP|>1$, has been studied in \cite{spohn} and \cite{moeller} (see also \cite{minlos1}, \cite{minlos2}  for some related spectral problems) but under some assumptions on the boson dispersion relation that change  the physical  phenomenon we are interested in.  In fact, in these papers,  the bosons are massive and their  energy dispersion relation is strictly subadditive (see \cite{moeller}). In  particular, in \cite{moeller}, it is proven that, for spatial dimension $d=3$, the fiber Hamiltonian has no groundstate whenever the infimum of its spectrum equals the infimum of its essential spectrum.   However, because of the assumptions above, this result does not apply to the model studied in this paper.

In the following, the spin of the electron is neglected, and the bosons are scalar.

\noindent \textbf{Acknowledgement} We thank an anonymous referee who pointed out the Remark on page 41.  At the time when this work was finished, W.D.R. was supported by the European Research Council and the Academy of Finland. A.P. is supported by  NSF grant DMS-0905988.

\section{Description of the model and result} \label{SectionII}
\resetequ

\subsection{Hilbert space}

The Hilbert space of pure states of the system is given by
\begin{equation}
\cH\,=\,L^2(\RR^3)\otimes\cF\,,
\end{equation}
where $\cF$ is the Fock space of scalar bosons,
\begin{equation} \label{eq-I-1a}
\cF \ := \ \bigoplus_{N=0}^\infty \cF^{(N)} 
\comma \hspace{6mm} \cF^{(0)} = \CC \, \Om \comma 
\end{equation}
with $\Om$ the vacuum vector, i.e., the state without any bosons, and the state space, $\cF^{(N)}$, of $N$ bosons is given by
\begin{equation} \label{eq-I-2}
\cF^{(N)} \ := \  
\cS_N \,  \fh^{\otimes\,N} 
\comma \hspace{6mm} 
N \geq 1 \,.
\end{equation}
Here the Hilbert space, $\fh$, of state vectors of a single  boson is given by
\begin{equation} \label{eq-I-2a}
\fh \ := \ L^2[ \RR^3 ] \,,
\end{equation}
and $\cS_N$ denotes symmetrization.
We introduce the usual
creation- and annihilation operators, $a^*_{\vk}$ and $a_{\vk}$, obeying the canonical 
commutation relations
\begin{eqnarray} \label{eq-I-6} 
[a^*_{\vk} \, , \, a^*_{\vk'}] 
\;= \;
[a_{\vk} \, , \, a_{\vk'}] \;& = &\; 0 \comma 
\\ \label{eq-I-7}
 [a_{\vk} \, , \, a^*_{\vk'}] 
\; &=& \; \, \delta (\vk - \vk') \comma 
\\ \label{eq-I-8}
 a_{\vk} \, \Om \; &= &\; 0 \comma 
\end{eqnarray}
for all $\vk, \vk' \in \RR^3$.
\\

\subsection{Fiber decomposition}

\noindent
We may write $\cH$ as a direct integral
\begin{equation}
\cH\,=\,\int^{\oplus}\,\cH_{\vec{P}}\,d^3P\,.
\end{equation}

\noindent
Given any $\vP\in\RR^3$, there is an isomorphism, $I_{\vP}$,
\begin{equation}
    I_{\vP}\,:\,\cH_{\vP}\,\longrightarrow\,\cF^{b}\,,
\end{equation}
from the fiber space $\cH_{\vP}$ to the Fock space $\cF^{b}$, acted upon by the annihilation- and
creation operators $b_{\vk}$, $b^*_{\vk}$,
where
$b_{\vk}$ corresponds to $e^{i\vk\cdot\vx}  a_{\vk}$, and
$b_{\vk}^*$ to $e^{-i\vk\cdot\vx} a_{\vk}^* $, and
with vacuum $\Omega_{f}:=I_{\vP}(e^{i\vP\cdot\vx})$.
To define $I_{\vP}$ more precisely, we consider a vector
$\psi_{(f^{(n)};\vP)}\in \cH_{\vP}$
with a definite total momentum describing an electron and $n$ bosons. Its wave function in the variables
$(\vx;\vk_1,\dots,\vk_n)$ is given by
\begin{equation}
 e^{i(\vP-\vk_1-\cdots-\vk_n)\cdot\vx}f^{(n)}(\vk_1,\dots,\vk_n)\,,
\end{equation}
where $f^{(n)}$ is totally symmetric in its $n$ arguments.
The isomorphism $I_{\vP}$ acts by way of
\begin{eqnarray}
    \lefteqn{I_{\vP}\big( e^{i(\vP-\vk_1-\cdots-\vk_n)\cdot\vx}f^{(n)}(\vk_1,\dots,\vk_n) \big)}
    \\
    &= &\frac{1}{\sqrt{n!}}\int \, d^3k_1\dots d^3k_n \,f^{(n)}(\vk_1,\dots,\vk_n)\,
    b_{\vk_1}^* \cdots
    b_{\vk_n}^*  \, \Omega_f \,.\nonumber
\end{eqnarray}

\subsection{Hamiltonians} \label{sec: hamiltonians}

We consider a non-relativistic particle moving in a medium of relativistic  bosons. The Hamiltonian of the system is given by
\begin{equation}
H:=\frac{1}{2} \vec{p}^{\,2}\,+\,g\phi(\rho_{\vx})\,+\,H^f\,,
\end{equation}
where:
\begin{itemize}
\item
 The operators $\vx\,,\,\vp$ describe the electron position and momentum, respectively;
 \item
$ H^f:=\,d\Gamma(\omega(|\vk|))$ (see Section \ref{sec: notation}), where $\omega(|\vk|):=|\vk|$, is the free field Hamiltonian. In physicist's notation
$$
H^f = \int d^3k  \, \str \vk \str  \, a^*_{\vk} a^{}_{\vk}.
$$
\item
 The real number $g$, $|g|>0$,  is a coupling constant.
\item
The interaction Hamiltonian is 
\begin{equation} \label{form-factor}
    \phi(\rho_{\vx}) \, := \, \int 
    d^3k\,\rho(\vk) \, (a^*_{\vk}\,e^{-i\vk\cdot\vx}\, + \, a_{\vk}\,e^{i\vk\cdot\vx}) \,, 
\end{equation}
where the form factor $\rho(\vk) \in \RR$ satisfies the following conditions
\begin{enumerate}
\item There is an ultraviolet cutoff $\Lambda$, i.e.\  $\rho(\vk)=0$ whenever $\str \vk\str > \Lambda$.
\item  The function $\rho$ is rotationally invariant, i.e., $\rho(\vk)=\rho(\str\vk\str)$, continously differentiable, $\rho \in C^1$. For expository convenience, when we will describe the decay mechanism in Theorem V.1, we will also assume that $\rho(\vk)  \neq 0$ for $0 < \str \vk \str < \Lambda $. Actually, this assumption is not necessary to state the main result of the theorem,  but simplifies the construction of the trial state in Eq. (\ref{eq: trial vector}) of Theorem V.1.
\item The following infrared regularity condition holds:
\begin{equation} \label{assumption on form factor}
|\rho(\vk)|\leq\cO(|\vk|^{\beta})\,, \qquad  \textrm{and}   \qquad   |\vnabla_{\vk}\rho(\vk)|\leq\cO(|\vk|^{\beta-1}),  \qquad \textrm{as} \, \,  \vk\to 0
\end{equation}
for an exponent $\beta>\betacrit$.  We believe that the critical value, $\beta=\betacrit$, is not optimal. From physical considerations,  the result concerning the instability of the mass shell should hold for any exponent $\beta  \geq -1/2$. For $\beta=-1/2$, the Hamiltonian describes  the interaction of the electron  with the quantized relativistic field with no infrared regularization.
\end{enumerate}

\end{itemize}
The operator $H$ is self-adjoint,  because  $\phi(\rho_{\vx})$ is an infinitesimal perturbation of $H^{0}:=H^f+\frac{\vec{p}^{\,2}}{2}$, and $\dom(H) = \dom(H^{0})$, i.e., the domains of self-adjointness coincide.
Since the Hamiltonian $H$ commutes with the total momentum, it preserves the
fiber spaces $\cH_{\vP}$, for all $\vP\in\RR^3$. Thus, we can write
\begin{equation}
    H\,=\,\int^{\oplus} H_{\vP}\,d^3P\,,
\end{equation}
where
\begin{equation}
    H_{\vP}\,:\,\cH_{\vP}\longrightarrow\cH_{\vP}\,.
\end{equation}
In terms of the operators $b_{\vk}$, $b^*_{\vk}$, and of the
variable $\vP$, the fiber Hamiltonian $H_{\vP}$ is given by
\begin{equation}\label{eq-fibHam}
    H_{\vP} \; := \; H^0_{\vP}+g\phi^b(\rho)\,,
\end{equation}
with
\begin{equation}
H^0_{\vP}:=\frac{\big(\vP-\vP^f\big)^2}{2}+ \;H^{f}\,,
\end{equation}
where, as operators on the fiber space $\cH_{\vP}$,
\begin{eqnarray}
    \vP^f &= &  \int d^3k\, \vk b^*_{\vk} \, b_{\vk} \, ,
    \\
    H^f& =&\,d\Gamma^b(\omega(|\vk|))= \int d^3k \, \omega(|\vk|) b^*_{\vk}  b_{\vk} \,,
\end{eqnarray}
and
\begin{equation}
    \phi^b(\rho) \, := \, \int \,
    d^3k\,\rho(\vk) \, (b^*_{\vk} 
    \, + \, b_{\vk}) \,.
\end{equation}
\\

\subsection{Result}\label{Section-III.1}
The absence of a mass-shell for $|\vP|>1$ is expressed by the following statement: The equation 
\begin{equation}\label{eq:III.20bis}
H_{\vP}\Psi_{\vP}=E_{\vP}\Psi_{\vP}
\end{equation}
has no  normalizable solution for any value of $E_{\vP}$ and for almost every $\vP\in\RR^3$, $|\vP|>1$. What we actually prove in this paper is the absence of \emph{regular} mass shells as formulated  in the theorem below (see also Figure 1).

More concretely, we address the  question whether, for a given region $I\times \Delta_I$ in the momentum-energy space (see (ii) below), there is an open interval $I_g$, $I_g\subset I$,  of size at least $\cO(|g|^{\gamma})$, $\gamma>0$,  where the mass shell exists, with $E_{\vP}\in \Delta_I$ and with the regularity property specified in the theorem. Recall that $\beta$ determines the infrared behaviour of the form factor $\rho$, see \eqref{assumption on form factor}.

\begin{theorem}\label{thm: main}
Assume that the form factor $\rho$ satisfies  \eqref{assumption on form factor},  with  $\beta>\betacrit$, and fix an interval $I$ of the form $I:=(1+\delta,\, \sigma), \delta>0,\sigma<\infty$ and a bounded interval $\Delta_I$. Fix constants $0< C_I, c_I  < \infty$ and exponents $0< \gamma < 1/4$ and $0<\ep< \gamma/4$. Then, there is a $g_{*}>0$ such that, for all $g$ satisfying $0< \str g \str < g_*$, {\bf{ the following is ruled out}}: \\

\noindent There exist normalizable solutions to equation (\ref{eq:III.20bis}), for all $\str \vP \str \in I_g$, such that:
\begin{itemize}
\item[(i)]  $I_g$ is an interval of length larger than $|g|^{\gamma/2}$ ($|I_g|\geq |g|^{\gamma/2}$).
\item[(ii)] $I_g\subset I$ and $E_{\vP} \subset \Delta_I$, for all $|\vP| \in I_g$.
\item [(iii)] For all $|\vP| \in I_g$,
$$
\Big\|\vnabla_{\vP}\Psi_{\vP,E_{\vP}}\Big\|\,<\,C_I.
$$

\item[(iv)]  For all $|\vP| \in I_g$,
$$ \left |E_{\vP}-(|\vP|-\frac{1}{2}) \right |>c_I \,|g|^{\gamma/4-\epsilon}.
$$
\end{itemize}
\end{theorem}

We note that it is an interesting open problem to understand whether single-particle states could emerge at the boundary of the energy-momentum spectrum, i.e.\ near $E_\vp= |\vP|-\frac{1}{2}$. Our results only rule out the existence of single particle states  whose energies are \emph{embedded} in the energy-momentum spectrum and with suitable regularity properties as far as their dependence on $\vP$ is concerned. \\

\noindent
{\bf{Remark}}
\noindent
In the following Theorems, Lemmas, and Corollaries, we  always assume  that the \emph{Main Hypothesis} in Section \ref{sec: main hypothesis} holds.  Furthermore,  $|g|$ ``sufficiently small" means  $0<|g|<g_*$,  where $g_*$ depends only on $I$, on $\Delta_I$, and on $\gamma$, but with the form factor $\rho$ and the ultraviolet cutoff $\Lambda$ kept fixed.
\\

\begin{figure}[h!]  \label{fig: result}
\psfrag{one}{$1$}
\psfrag{deltaone}{ $\delta$}
\psfrag{mom}{ $\str \vP \str$}
\psfrag{ene}{ $E$}
\psfrag{free}{ $\frac{1}{2}\str \vP\str^2$}
\psfrag{lineair}{ $\str \vP \str-\frac{1}{2}$}
\psfrag{gap}{ $\sim \str g \str^{\gamma/4-\epsilon}$}
\hspace{-1cm}
\includegraphics[width = 16cm, height=6cm]{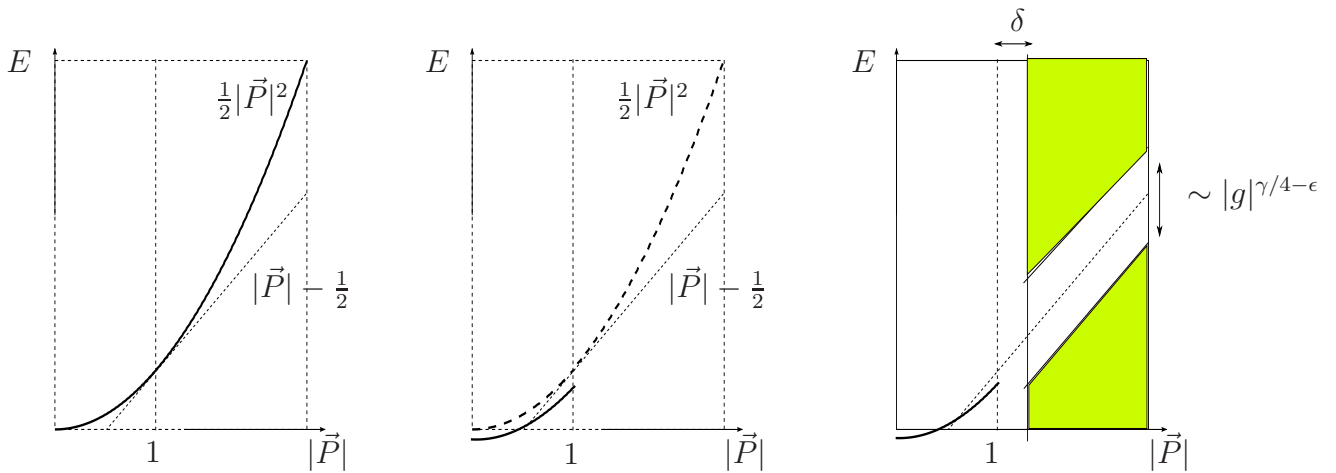}   
\caption{\footnotesize{The joint energy-momentum spectrum. By the rotation symmetry, it suffices to plot the $(E, \str \vP \str)$- plane.  In the leftmost figure, we have drawn the spectrum of the uncoupled system.  The parabola $\frac{1}{2}\str \vP \str^2$ (in boldface) is the mass shell and the spectrum  lies above the three-dimensional surface consisting of $\frac{1}{2}\str \vP \str^2$, for $\str\vP|<1$, and $\str \vP \str -1/2$, for
$\str \vP \str >1$. Hence, for $\str \vP \str >1$, the mass shell is \emph{embedded} in the continuum. 
 In the middle figure, we represent the situation when the coupling is switched on, \emph{according to formal perturbation theory}. The mass shell has dissappeared (drawn as a dashed line) for $\str \vP \str >1$. For $\str \vP \str <1$, the mass shell persists but gets deformed (mass renormalization). 
  In the rightmost figure, we represent what is know rigorously:  a regular mass shell is excluded in the coloured area (result of the present paper) and there is a renormalized mass shell for small $\str P \str $ (earlier works, see Section \ref{sec: intro} ). }
  }
\end{figure}

\subsubsection{Main ingredients of the proof}
\begin{itemize}
\item[(a)]
If $\Psi_{\vP,E_{\vP}}$ existed with properties (i)-(iv) above, and $||\vnabla E_\vP|-1|>\frac{3}{2}|g|^{\gamma/3}$, then 
\begin{equation}
\|\Psi_{\vP}^0-\Psi_{\vP,E_{\vP}}\|\leq \cO(|g|^{(1-2\gamma)/6}),
\end{equation}
where $\Psi_{\vP}^{0}$ is the bare one-particle state, (i.e., $\Psi_{\vP}^{0}=\Omega_{f})$, and
\begin{equation}
\big|\frac{\vP^2}{2}-E_{\vP}\big|\leq \cO(|g|^{(1-2\gamma)/6})\,.
\end{equation}
\item[(b)]
If  $\Psi_{\vP,E_{\vP}}$, as in (a), existed then it could decay into a state consisting of an unperturbed single particle state and a boson with momentum $\vk$ in a region of momentum space away from the ray $\{\lambda \vP\, |\, 0<\lambda \leq \infty \}$.
\item[(c)]
If $||\vnabla E_\vP|-1| \leq\frac{3}{2}|g|^{\gamma/3}$ then $|E_{\vP}-(|\vP|-\frac{1}{2})|<const\, |g|^{\gamma/4}$. In other words, a mass shell with group velocity close to one, necessarily lies near  the boundary of the energy momentum spectrum.   
\end{itemize}

\subsection{Notation} \label{sec: notation}

Here is a list of notations used in subsequent sections.

\begin{enumerate}
\item Given any vector $\vec{u}\in\RR^3$, $\hat{u}:=\frac{\vec{u}}{|\vec{u}|}$.

\item $\cF_{fin}$ is the dense subspace of $\cF$ obtained as the span of vectors containing  finitely many bosons.

\item $\one_{(a\,,\,b)}(\vk)$ is the characteristic function of the set $$\{\vk\in\RR^3\,\,:\,\,|\vk|\in(a\,,\,b)\}\,.$$

\item For any function $w\in \fh$, $\|w\|_{2}$ is the corresponding $L^2$-norm.

\item $d\Gamma(A)$ is the second quantization of an operator $A$ acting on $\fh$;
$d\Gamma(A)$ is an operator on $\cF$. Analogously, $d\Gamma^b(A)$ is defined
on $\cF^b$.

\item We define the (boson) number operators by
$N:=d\Gamma(\one(\vk))$ and  $N^b:=d\Gamma^b(\one(\vk))$, where $\one(\vk)$ is the identity operator on $L^2(\RR^3;d^3k)$.

\item We use the notation
$$
a^*(f_\vx) := \int d^3k \,  f_\vx(\vk)  a^*_\vk, \qquad     a(f_\vx) := \int d^3k  \,  \overline{f_\vx(\vk) }  a_\vk
$$
for smeared creation/annihilation operators, depending also on  the (electron) position $\vx$. 
\item Expressions like  $(\vP,E_{\vP}) \in I_g\times \Delta_I $ are interpreted as follows: $\vP\in \RR^3$ with $|\vP|\in I_g$,  and $E_{\vP}\in \Delta_{I}$.

\end{enumerate}

\subsection{Structure of the paper}

In Section \ref{Section-III} below, we state a \emph{Main Hypothesis} (Section \ref{subsec-II.1}).
The upshot of our analysis is Theorem \ref{theo-VI.4} in Section \ref{Section-VI}. This theorem describes the possible location of a mass shell, under the assumption that the \emph{Main Hypothesis}  holds true. 
In other words, the implication 
\begin{equation}
\textrm{\emph{Main Hypothesis}}   \mathop{\Longrightarrow}\limits_{\footnotesize{\left.\begin{array}{c} \textrm{Assumptions} \\ \textrm{ in Section \ref{sec: hamiltonians} }  \end{array} \right.  }}   \textrm{Theorem \ref{theo-VI.4} }
\end{equation}
is our main result, and this implication gives rise to Theorem \ref{thm: main}. 

In the remainder of Section \ref{subsec-II.1}, we state some immediate consequences of the \emph{Main Hypothesis}, and in Section  \ref{subsec-II.2},  we put the technical tools in place. 
Section \ref{SectionIII.2}  contains a rather detailed description of the strategy of our proofs. 
The proofs themselves are presented in Sections \ref{Section-V} and \ref{Section-VI}. 
 An appendix contains the proofs of some preliminary results used in Section   \ref{Section-V}.
\section{Strategy of the proof}\label{Section-III}
\resetequ

\subsection{Main Hypothesis and key properties\label{subsec-II.1}}

The proof of our result, Theorem \ref{thm: main}, is by contradiction. We will assume that  a regular mass shell exists, and subsequently, we derive that it cannot be located anywhere else than near the boundary of the energy-momentum spectrum.  Our assumption is stated in Section \ref{sec: main hypothesis} below and it will be referred to as the \emph{Main Hypothesis}.  Throughout the rest of the paper, we assume that the \emph{Main Hypothesis} holds. 
In Section \ref{sec: properties}, we derive some consequences of the \emph{Main Hypothesis}, namely Properties P1, P2 and P3. 

\subsubsection{Main Hypothesis} \label{sec: main hypothesis}

Let $R$ be a rotation matrix in $\RR^3$ and $U(R)$ the unitary operator implementing the transformation 
\begin{equation}
b_{\vk}\quad\rightarrow\quad b_{R^{-1}\vk}\,=:\,U^{*}(R)\,b_{\vk}\,U(R)\,.
\end{equation}
The identity
\begin{equation}
U^{*}(R)\,H_{R\vP}\,U(R)\,=\,H_{\vP}\,,
\end{equation}
implies that if $\Psi_{\vP,E_{\vP}}$ is a normalized eigenvector of $H_{\vP}$ with eigenvalue $E_{\vP}$ then
$U(R)\Psi_{\vP,E_{\vP}}$ is an eigenvector of $H_{R\vP}\,$ with the same eigenvalue, i.e.,
\begin{equation}
H_{R\vP}U(R)\Psi_{\vP,E_{\vP}}\,=\,E_{\vP}\,U(R)\Psi_{\vP,E_{\vP}}\,.
\end{equation}
In particular, the existence of an eigenvector, $\Psi_{\vP,E_{\vP}}$, of $H_{\vP}$ for all $\vP$ in a given direction, $\hat{P}$, yields a mass shell with energy function $E_{\vP}\equiv E_{|\vP|}$. 
\\

\noindent
\textbf{Main hypothesis}: \emph{We temporarily assume that single-particle states, $\Psi_{\vP,E_{\vP}}$, exist, i.e.,}
\begin{equation}
H_{\vP}\Psi_{\vP,E_{\vP}}\,=\,E_{\vP}\,\Psi_{\vP,E_{\vP}},\, \qquad \norm  \Psi_{\vP,E_{\vP}} \norm=1\,,
\end{equation}
\emph{such that the vector $\Psi_{\vP,E_{\vP}}$ is differentiable in $\vP$ with }
\begin{equation}\label{eq:II.4}
\Big\|\vnabla_{\vP}\Psi_{\vP,E_{\vP}}\Big\|\,<\,C_{I}\,,
\end{equation}
\emph{where the constant $C_I< \infty$, for all $\vP$ such that $|\vP|\in I_g\subset I$ and for $E_{\vP}\in \Delta_{I}$, where  $\Delta_I$ is a bounded interval. Here, $I_g$ is an open interval, $|I_g|>|g|^{\gamma/2}$, and $I:=(1+\delta, \sigma)\, , \,\sigma-1>\delta>0$.} 
\\

From the assumption in Eq. (\ref{eq:II.4}), the following properties follow  for $|\vP| \in I_g\,,\,E_{\vP}\in \Delta_{I} $.
\\

\subsubsection{Properties (P1), (P2) and (P3)} \label{sec: properties}
\begin{itemize}
\item[(P1)]
$E_{\vP}=E_{|\vP|}$ is differentiable and the Feynman-Hellman formula holds
\begin{equation}\label{eq:II.5.1}
\vnabla E_{\vP}\,=\,(\Psi_{\vP,E_{\vP}}\,,\,(\vP-\vP^f)\,\Psi_{\vP,E_{\vP}})\,.
\end{equation}
The expression on the R.H.S. (right-hand side)  of (\ref{eq:II.5.1}) is continuous in $\vP$. Thus $\vnabla E_{\vP}$ is a continuous function of $\vP$. Moreover, $|\vnabla E_{\vP}|<C'_{I}$ for some $C'_I <\infty$, and, because of rotation invariance, $\vnabla E_{\vP}$ and $\vP$ are colinear.
\item[(P2)]
For some $0<C_{I}''< \infty$,
\begin{equation}\label{prop-P2}
\Big|\frac{\partial_{|\vP|}^2 E_{\vP}}{\partial |\vP|^2}\Big|\leq C_{I}''\,.
\end{equation}
Starting from the derivative of the R.H.S. of (\ref{eq:II.5.1}), this bound can be easily obtained using  (\ref{eq:II.4})  and that $H^0_{\vP}$ is  $H_{\vP}$-bounded.
\item[(P3)]
From $H^f=\,d\Gamma^b(|\vk|)$ and Eq. (\ref{eq:II.4}), it follows that
\begin{equation}\label{eq:II.6}
|\vnabla_{\vP}\,(\Psi_{\vP,E_{\vP}}\,,\,d\Gamma^b (\one_{(\frac{1}{n+1},\frac{1}{n})}(\vk))\,\Psi_{\vP,E_{\vP}})|\,\leq \cO(n\,C_I\, [(\sup_{\vP\in I}|E_{\vP}|)+1]\,; 
\end{equation}
here we use the  inequality $$\|d\Gamma^b (\one_{(\frac{1}{n+1},\frac{1}{n})}(\vk))\,\psi\|\leq(n+1)\|H^f\,\psi\|\,,\quad \forall \psi\in \text{Dom}(H^f)\,,$$ and that $H^f$ is $H_{\vP}$-bounded.
\end{itemize}
\subsection{Technical Tools }\label{subsec-II.2}

\noindent
We will use two different virial arguments to expand $\Psi_{\vP,E_{\vP}}$ in the coupling constant $g$, $|g|\ll 1$. For this purpose, we must introduce single-particle ``wave packets", $\Psi_{f_{\vec{Q}}^g}$, defined below.

\begin{itemize}
\item[($\mathcal{I}$)]
\emph{Single-particle ``wave packets", $\Psi_{f_{\vec{Q}}^g}$, and the interval $I'_g$.}

For $|g|$ small enough, we define the open interval $I'_g$ such that
\begin{equation}\label{I-constraint}
(|I_g|>)|I'_g|>4|g|^{\gamma}\,,
\end{equation}
with the property that 
\begin{equation}\label{II-constraint}
|\vec{Q}|+|\vec{z}|\in I_g\,,
\end{equation}
for all $|\vec{Q}|\in I'_g$ and for all $\vec{z}$ such that $|\vec{z}|<|g|^{\gamma}$.

\noindent
We consider single-particle ``wave packets", $\Psi_{f_{\vec{Q}}^g}$,  with wave function, $f_{\vec{Q}}^g$,  centered around vectors $\vec{Q}$,  $|\vec{Q}|\in I'_g$. The vector  $\Psi_{f_{\vec{Q}}^g}$ is defined by
\begin{equation}\label{eq:II.8}
\Psi_{f_{\vec{Q}}^g}\,:=\,\int f_{\vec{Q}}^g(\vP)\,\Psi_{\vP,E_{\vP}}d^3P
\end{equation} 
where
$f_{\vec{Q}}^g(\vP):=\cR(\frac{|\vP|-|\vec{Q}|}{|g|^{\gamma}})\cA(\frac{\theta_{\hat{QP}}}{|g|^{\gamma}})$, $\theta_{\hat{QP}}$ is the angle between $\vec{Q}$ and $\vec{P}$, and $\cR(z)$, $\cA(\theta)$ are defined as follows.

\noindent
1) $\cR(z)$, $z\in\RR$,  is non-negative, smooth and compactly supported in  the interval $(-1,1)$, $\cR(z)=1$ for $z\in (-\frac{1}{2},\frac{1}{2})$, 

\noindent
2) $\cA(\theta)$, $\theta \in \RR$, is non-negative, smooth and compactly supported in  the interval $(-1,1)$,  $\cA(\theta)=1$ for $ \theta \in (-\frac{1}{2},\frac{1}{2})$.  Therefore, the angular restriction  $$\hat{P}\cdot\hat{P'}\geq \cos(|g|^{\gamma})$$  holds for any $\vP\,,\,\vP'\in\,supp\,f_{\vec{Q}}^g$.

\noindent
3) Since $|\vec{Q}|>1$, it follows  from the definitions of $\cR(z)$ and $\cA(\theta)$  that:
 $$|\vnabla_{\vP}f_{\vec{Q}}^g(\vP)|\leq \cO(|g|^{-\gamma})\,,$$  for any $\vP\in\,supp\,f_{\vec{Q}}^g$. 
\\

\item[($\mathcal{II}$)]
\emph{Multi-scale virial argument on the Hilbert space $\cH$ for the Hamiltonian $H$.}

We define dilatation operators on the one-particle space $\fh$,  constrained to a suitable range of frequencies and to a suitable angular sector around a direction $\hat{u}$. We introduce the conjugate operator 
\begin{equation}\label{eq:III.1}
D_{n,\perp}^{\hat{u},\hat{Q}}\,:=\,d\Gamma(d_{n,\perp}^{\hat{u},\hat{Q}})\,,
\end{equation}  
with 
\begin{equation}\label{eq:II.10}
d_{n,\perp}^{\hat{u},\hat{Q}}\,:=\,\chi_{n}(|\vk|)\xi_{\hat{u}}^g(\hat{k})\frac{1}{2}\,(\vk_{\perp}\cdot\,i\vnabla_{\vk_{\perp}}+i\vnabla_{\vk_{\perp}}\cdot\,\vk_{\perp})\xi_{\hat{u}}^g(\hat{k})\chi_{n}(|\vk|)\,,
\end{equation}
where:

\begin{itemize} 
\item [(a)]
$\vk_{\perp}$ is the component of the vector $\vk$ orthogonal to $\vec{Q}$, i.e.,  $\vk_{\perp}:=\vk-\frac{\vk\cdot\vec{Q}}{|\vec{Q}|^2}\,\vec{Q}$;

\noindent
$\chi_{n}(|\vk|)$, $n\in\NN$, are  non-negative, $C^{\infty}(\RR^+)$ functions with the properties:
\begin{itemize}
\item[(i)]
$\chi_{n}(|\vk|)=0$ for $|\vk|\leq\frac{1}{2(n+1)}$ and for $|\vk|\geq\frac{3}{2n}$;
\item[(ii)]
$\chi_{n}(|\vk|)=1$ for $\frac{1}{n+1}\leq|\vk|\leq\frac{1}{n}$;
\item[(iii)]
$|\chi_{n}'(|\vk|)|\leq C_{\chi}\,n$, for all $n\in\NN$, where the constant $C_{\chi}$ is independent of  $n$.
\end{itemize}
 
\item [(b)]
\noindent
 $\xi_{\hat{u}}^g(\hat{k})$ (see Figure 2), $0\leq \xi_{\hat{u}}^g(\hat{k})\leq 1$,  is a  smooth function with support in the $g$-dependent
cone
 \begin{equation}\label{eq:function-cone1}
\mathcal{C}_{\hat{u}}:= \{\hat{k}\,\,:\,\,\hat{k}\cdot\hat{u}\geq \cos(|g|^{\gamma})\}\,,
 \end{equation}
 such that:
 \begin{itemize}
\item[i)]
 \begin{equation}
 \xi_{\hat{u}}^g(\hat{k})=1\quad\text{for}\quad\{\hat{k}\,\,:\,\,\hat{k}\cdot\hat{u}\geq \cos(\frac{1}{2}|g|^{\gamma})\}\,;
 \end{equation}
\item[ii)] 
 \begin{equation}\xi_{\hat{u}}^g(\hat{k})=0\quad\text{for}\quad\{\hat{k}\,\,:\,\,\hat{k}\cdot\hat{u}< \cos(|g|^{\gamma})\}\,;
 \end{equation}
\item[iii)] 
\begin{equation}
|\partial_{\theta_{\hat{ku}}}\xi_{\hat{u}}^g(\hat{k})|\leq C_{\xi}\,|g|^{-\gamma}\,,\label{eq:function-cone4}
\end{equation}
where $\theta_{\hat{ku}}$ is the angle between $\hat{k}$ and $\hat{u}$, and the constant $C_{\xi}$ is independent of $g$.
\end{itemize}
\end{itemize}
\begin{figure} \label{fig: singlecone}
\includegraphics[width=1.0\textwidth,height=0.4\textheight]{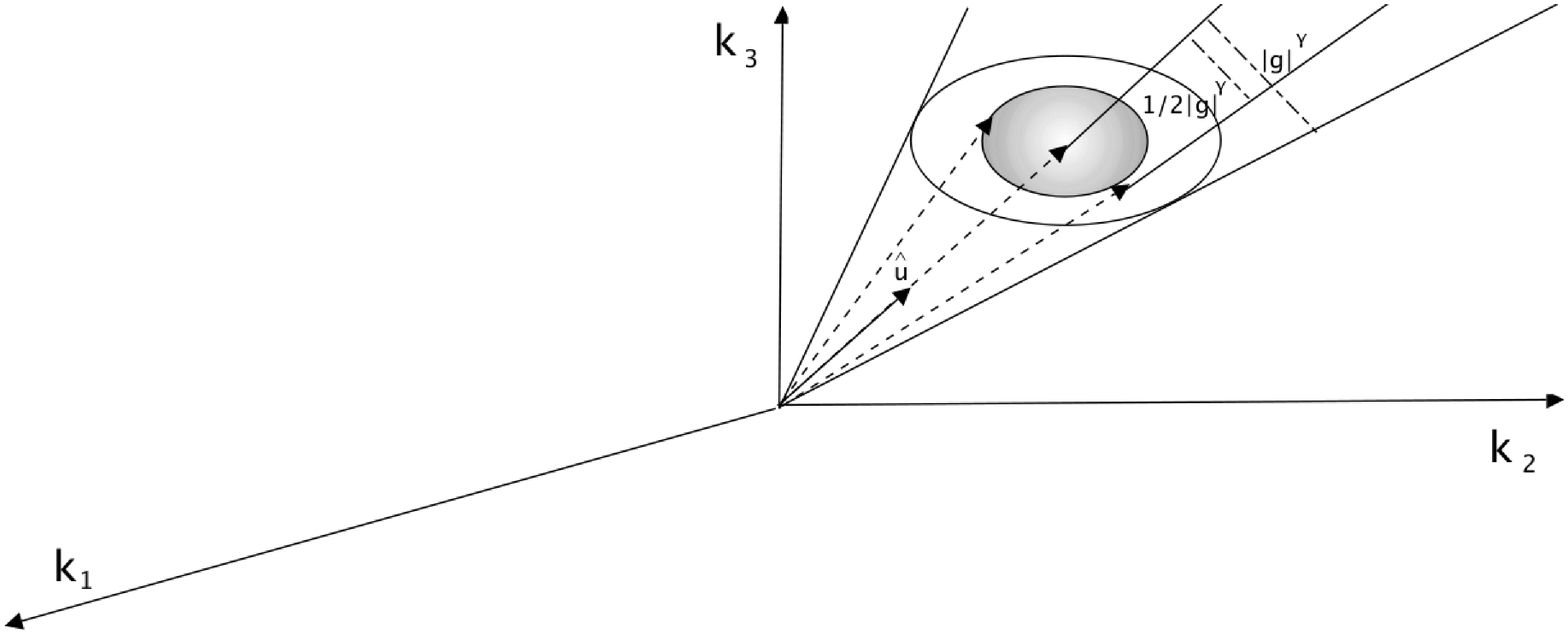}
\caption{The cone $\mathcal{C}_{\hat{u}}$ corresponding to the support of the smooth characteristic function $\xi_{\hat{u}}^g(\hat{k})$.}
\end{figure}

We also define 
\begin{equation}\label{eq:II.12}
d_{n}^{\hat{u}}\,:=\,\frac{1}{2}\chi_{n}(|\vk|)\xi_{\hat{u}}^g(\hat{k})[\,\vk\cdot\,i\vnabla_{\vk}+i\vnabla_{\vk}\cdot\,\vk\,]\chi_{n}(|\vk|)\xi_{\hat{u}}^g(\hat{k})\,,
\end{equation}
and we introduce the second quantized operator
\begin{equation}\label{eq:III.6}
D_{n}^{\hat{u}}\,:=\,d\Gamma(d_{n}^{\hat{u}})\,.
\end{equation}

Later on in the paper, when we implement the virial argument,  we will make use of the creation/annihilation operators
$$ a^*(id_n^{\hat{u}}\rho_\vx)=\int d^3k \, i (d_n^{\hat{u}}\rho_\vx)(\vk) a^*_{\vk} $$
$$ a(id_n^{\hat{u}}\rho_\vx)=\int d^3k  \,  \overline{ i (d_n^{\hat{u}}\rho_\vx)(\vk)} a_{\vk}, $$
and, analogously, $a(id_{n,\perp}^{\hat{u},\hat{Q}} \rho_{\vx})$, $a^{*}(id_{n,\perp}^{\hat{u},\hat{Q}} \rho_{\vx})$.  
In Lemma \ref{lem-III.1}, we show that  the vector $\Psi_{f_{\vec{Q}}^g}$ belongs to the form domain of these operators.
\item[ ($\mathcal{III}$)]
\emph{Virial argument in each fiber space $\cH_{\vP}$.}

\noindent
Here we consider
$$
D^b_{\frac{1}{\kappa},\kappa}\,:=\,d\Gamma^b(d_{\frac{1}{\kappa},\kappa})\,
$$
as  the conjugate operator, 
where 
\begin{equation}  \label{eq:III.7}
d_{\frac{1}{\kappa}, \kappa}\,:=\,\chi_{[\frac{1}{\kappa}, \kappa]}(|\vk|)\frac{1}{2}\,(\vk\cdot\,i\vnabla_{\vk}+i\vnabla_{\vk}\cdot\,\vk)\chi_{[\frac{1}{\kappa},\kappa]}(|\vk|)\,.
\end{equation}
$\chi_{[\frac{1}{\kappa},\kappa]}(|\vk|)$, $\infty>\kappa>\max\{\Lambda, 1\}$,  are  non-negative, $C^{\infty}(\RR^+)$ functions with the properties:
\begin{itemize}
\item[(i)]
$\chi_{[\frac{1}{\kappa},\kappa]}(|\vk|)=0$ for $|\vk|\geq 2\kappa $, $|\vk|\leq \frac{1}{2\kappa}$;
\item[(ii)]
$\chi_{[\frac{1}{\kappa},\kappa]}(|\vk|)=1$ for $\frac{1}{\kappa} \leq|\vk|\leq\kappa$;
\item[(iii)]
for some $C>0$,  $|\chi'_{[\frac{1}{\kappa},\kappa]}(|\vk|)|<C \kappa$;
\end{itemize}

Analogously to $a^{*}(id_{n}^{\hat{u}} \rho_{\vx}), a(id_{n}^{\hat{u}} \rho_{\vx})$, we will use
$$ b^*(id_{\frac{1}{\kappa},\kappa}\,\rho)=\int d^3k  \, (id_{\frac{1}{\kappa},\kappa}\,\rho)(\vk)  b^*_{\vk}$$ 
$$b(id_{\frac{1}{\kappa},\kappa}\,\rho)= \int d^3k  \,  \overline{(id_{\frac{1}{\kappa},\kappa}\,\rho)(\vk)}  b_{\vk}.$$


\end{itemize}
We will also consider the $g$-dependent  cones (see Figure 3),
\begin{equation}\label{eq:cone}
\mathcal{C}_{\hat{P}}^a:=\{\vk\,\,: \,\,|\hat{k}\cdot\hat{P}|\leq \cos(a|g|^{\gamma/8})\}\,,
\end{equation}
and use the smooth functions  $\xi_{\mathcal{C}_{\hat{P}}^a}^g(\hat{k})$, $a=\frac{1}{2},2$,  defined below.

\noindent
The functions $\xi_{\mathcal{C}_{\hat{P}}^a}^g(\hat{k})$, $0\leq \xi_{\mathcal{C}_{\hat{P}}^a}^g(\hat{k})\leq 1$, are chosen such that
\begin{itemize}
\item[i)]
\begin{equation}\label{eq:cone2-1}
\xi_{\mathcal{C}_{\hat{P}}^a}^g(\hat{k})=1\quad\text{for}\quad\{\hat{k}\,\,:\,\,|\hat{k}\cdot\hat{P}|\leq \cos(2a|g|^{\gamma/8})\}\,;
\end{equation}
\item[ii)]
\begin{equation}\label{eq:cone2-2}
\xi_{\mathcal{C}_{\hat{P}}^a}^g(\hat{k})=0\quad\text{for}\quad\{\hat{k}\,\,:\,\,|\hat{k}\cdot\hat{P}|> \cos(a|g|^{\gamma/8})\}\,;
\end{equation}
\item[iii)]
\begin{equation}\label{eq:cone2-3}
|\partial_{\theta_{\hat{kP}}}\xi_{\mathcal{C}_{\hat{P}}^a}^g(\hat{k})|\leq C_{\xi}\,|g|^{-\gamma/8}\,,
\end{equation}
 for a  constant $C_{\xi}$ independent of $g$, where  $\theta_{\hat{kP}}$  is the angle between $\hat{k}$ and $\hat{P}$. 
\end{itemize}
\begin{figure}\label{fig: doublecone}
\includegraphics[width=0.8\textwidth,height=0.4\textheight]{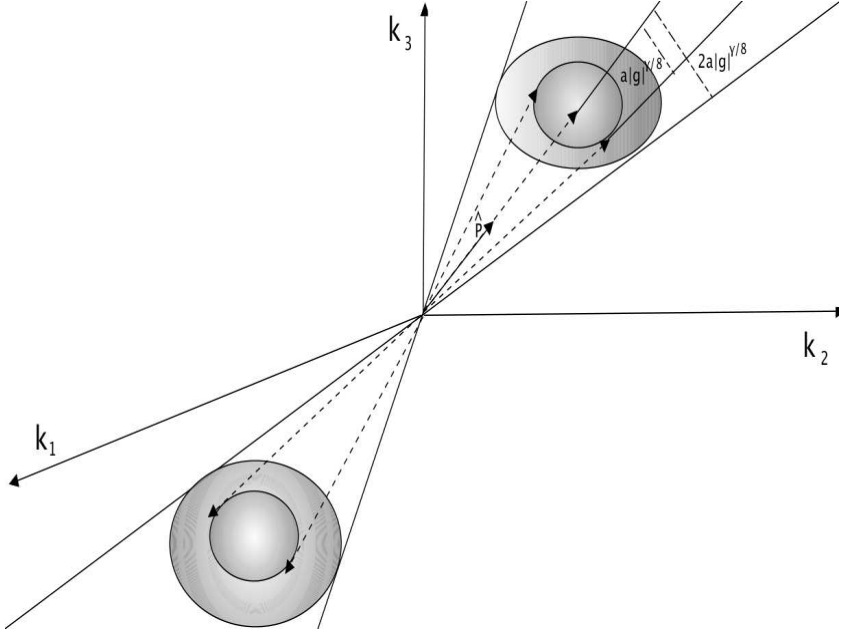}
\caption{The double cone $\mathcal{C}_{\hat{P}}^{a}$ is the complement in $\RR^3$ of the inner double cone around $\hat{P}$ of angular width $a|g|^{\gamma/8}$.}
\end{figure}
\newpage
\subsection{Description of  strategy}\label{SectionIII.2}

To exclude the existence of eigenvalues $E_{\vP}$, $|\vP|>1$,  we elaborate on  an argument introduced in \cite{FP}. 
The idea of the proof is as follows. One assumes that an
eigenvector $\Psi_{\vP,E_{\vP}}\in \cH_{\vP}$ of $H_{\vP}$ exists,  
for some energy $E_{\vP}$ in a compact set. Then, using a multiscale virial argument, one intends to prove  that
\begin{equation}
(\Psi_{\vP,E_{\vP}}, N^b \Psi_{\vP,E_{\vP}})\leq \cO(g^2)\,,\label{eq:III.24bis}
\end{equation}
where $N^{b}$ is the boson number operator in the fiber spaces. The
multiscale virial argument involves the dilatation operators
\begin{equation}
d_n\,:=\,\frac{1}{2}\chi_{n}(|\vk|)[\,\vk\cdot\,i\vnabla_{\vk}+i\vnabla_{\vk}\cdot\,\vk\,]\chi_{n}(|\vk|)\,,
\end{equation}
on the one-particle space $\fh$, where $\chi_{n}$ is a
suitable smooth approximation to the characteristic function of the interval
$[\frac{1}{n+1}\,,\,\frac{1}{n}]$ contained in the positive
frequency half axis, $n=1,2,...$. 
After introducing the second quantized dilatation operators $D_n^b:=d\Gamma^b(d_n)$, one starts
from the \emph{formal} virial
identity 
\begin{equation}\label{eq:III.25}
0=(\Psi_{\vP,E_{\vP}}\,,\,i[H_{\vP}\,,\,D_n^b]\,\Psi_{\vP,E_{\vP}})
\end{equation} 
to establish the scale-by-scale inequality below, in a rigorous way:
\begin{equation}\label{eq:III.25-1}
(\Psi_{\vP,E_{\vP}}\,,\,N^{b}_{n}\Psi_{\vP,E_{\vP}})\leq \cO(g^2n^2\||\vk|^{\beta}\,\one_{(\frac{1}{2(n+1)},\frac{3}{2n})}(\vk)\|_2^2)\,,
\end{equation}
where $N^{b}_{n}:=\int\,d^3k\,a^*(\vk)\chi_{n}^2(|\vk|)a(\vk)$, $n=1,2,3,\dots$. If (\ref{eq:III.25-1}) holds true, for sufficiently large values of the exponent $\beta$ in the form factor $\rho$, one can sum over $n$ and conclude that $(\Psi_{\vP, E_{\vP}}\,,\,N^{b}\Psi_{\vP,E_{\vP}})\leq \cO(g^2)$.
Next, the eigenvalue equation (\ref{eq:III.20bis}) and the inequality in Eq. 
(\ref{eq:III.24bis}) can be combined to conclude that the vector $\Psi_{\vP,E_{\vP}}$ and
the eigenvalue $E_{\vP}$ must fulfill the following estimates:
\begin{equation} \label{eq:asymptotic1}
\|\Psi_{\vP,E_{\vP}}-\Psi_{\vP}^{0}\|^2\leq\cO(g^2)\,,
\end{equation}
where $\Psi_{\vP}^{0}:=\Omega_f$ is the unperturbed eigenstate, and
\begin{equation} \label{eq:asymptotic2}
|E_{\vP}-\frac{\vP^2}{2}|\leq\cO(g^2)\,.
\end{equation}
This result would imply that putative eigenvalues of $H_{\vP}$ lie in an $\cO(g^2)$-neighborhood of
the eigenvalue of the Hamiltonian $H^{g=0}_{\vP}$. 

\noindent
Then the argument  proceeds with the construction of suitable trial states of the type
\begin{equation}
\eta_\vP\,:=\,\int\,d^3k\,\,\frac{1}{\epsilon^{\frac{1}{2}}}\,h\big(\frac{(\vP-\vk)^2/2+|\vk|-E_{\vP}}{\epsilon}\big)b^{*}_{\vk}\,\Psi_{\vP}^{0}\,,\label{eq:trial}
\end{equation}
where $\epsilon>0$ and $h(z)\in C_0^{\infty}(\RR)$, $h(z)\geq0$. 
One then exploits the identity
\begin{equation}\label{eq:III.23 intro}
(\eta_\vP\,,\,(H_{\vP}-E_{\vP})\Psi_{\vP,E_{\vP}})=0\,
\end{equation}
that must hold true if $\Psi_{\vP,E_{\vP}}$ is an eigenvector of $H_{\vP}$.
Starting from Eqs. (\ref{eq:asymptotic1})-(\ref{eq:asymptotic2}),  and using that  the equation
\begin{equation}\label{eq:root}
(\vP-\vk)^2/2+|\vk|-E_{\vP}=0
\end{equation}
has  solutions  for $|\vP|>1$,  provided $|g|$ is small enough, one arrives at a contradiction,  for $\epsilon$ and $|g|$ small enough. 

\noindent
However,  the procedure just  outlined (mimicking the treatment of atomic resonances in \cite{FP}) will not  work without some important modifications. We will therefore implement analogous, but more elaborate strategy.  

\noindent
The first problem ecountered is that  we cannot control the expectation value $$(\Psi_{\vP,E_{\vP}}, N^b \Psi_{\vP,E_{\vP}})$$ by a multiscale virial argument in the fiber space $\cH_{\vP}$,  because of the term $(\vP^f)^2$ in $H_{\vP}$.  The commutator of $(\vP^f)^2$  with $d\Gamma^b (d_{n})$, formally given by
\begin{equation}
\vP^f\cdot d\Gamma^b(\chi_n^2(\vk)\vk)+ d\Gamma^b(\chi_n^2(\vk)\vk)\cdot \vP^f\,,
\end{equation}
cannot be controlled in terms of  the commutator of $H^f$ with $d\Gamma^b (d_{n})$ . Consequently,  the estimate in Eq. (\ref{eq:III.25-1}) cannot be justified starting from the virial identity in Eq. (\ref{eq:III.25}).

\noindent
At the price of limiting our analysis to \emph{regular} mass shells (see \emph{Main Theorem} in Section \ref{Section-III.1}), this problem can be circumvented by implementing a multiscale virial argument in the full Hilbert space, by using single-particle ``wave packets"  rather than fiber eigenvectors, i.e., vectors in $\cH$ of the type
\begin{equation}
\Psi_{f}\,:=\,\int f(\vP)\,\Psi_{\vP,E_{\vP}}d^3P\,,
\end{equation}
where $f(\vP)$ is a smooth function  with support in $I_g$ (the region of  momenta for which an eigenstate was assumed to exist).  In practice, we choose  $f =f_{\vec{Q}}^g$  to be sharply peaked around a given momentum $\vec Q$, see definition below \eqref{eq:II.8}.  In the full Hilbert space, we can essentially mimick the treatement of atomic resonances to derive the following result (see  Section \ref{Section-V}).

\begin{blanktheorem}[\textrm{\ref{theo-III.2}}]
For $|g|$ sufficiently small,
\begin{equation}\label{eq:III.12a-intro}
\frac{( \Psi_{f_{\vec{Q}}^g}\,,\,N_{n,\mathcal{C}_{\hat Q}^{1/2}}\,\Psi_{f_{\vec{Q}}^g})}{( \Psi_{f_{\vec{Q}}^g}\,,\,\Psi_{f_{\vec{Q}}^g})}\leq\,\cO(g^{2(1-2\gamma)}n^2\||\vk|^{\beta}\,\one_{(\frac{1}{2(n+1)},\frac{3}{2n})}(\vk)\|_2^2) \,,
\end{equation}
where $N_{n,\mathcal{C}_{\hat Q}^{1/2}}:=d\Gamma(\chi_{n}^2(|\vk|)\xi_{\mathcal{C}_{\hat Q}^{1/2}}^{g\,2}(\hat{k}))$ and $|\vec{Q}|\in I'_g$.  

Furthermore, if for all  $\vP\in supp f_{\vec{Q}}^g$ the inequality $$||\vnabla E_{\vP}|-1|>|g|^{\gamma/3}$$ holds true, then
\begin{equation}\label{eq:III.14a-intro}
\frac{(\Psi_{f_{\vec{Q}}^g}\,,\,N_{n}\,\Psi_{f_{\vec{Q}}^g})}{( \Psi_{f_{\vec{Q}}^g}\,,\,\Psi_{f_{\vec{Q}}^g})}\leq  \, \cO(g^{2(1-2\gamma)}n^2\||\vk|^{\beta}\,\one_{(\frac{1}{2(n+1)},\frac{3}{2n})}(\vk)\|_2^2)\,,
\end{equation}
where $N_{n}:=d\Gamma(\chi_{n}^2(|\vk|)$. 

\noindent
The constants in (\ref{eq:III.12a-intro}), (\ref{eq:III.14a-intro})  can be chosen uniformly in $\vec{Q}$,  $|\vec{Q|}\in I'_g\subset I$ ($I'_g$ is defined in Section \ref{subsec-II.2},  Eqs. (\ref{I-constraint}),(\ref{II-constraint})). They only depend on  $I$ and on $\Delta_I$.
\end{blanktheorem}

By exploiting the $g$-dependence of the wavefunctions $f_{\vec{Q}}^g$  and  the assumption on the regularity in $\vP$ of $\Psi_{\vP,E_{\vP}}$, one can convert a bound for the number operator $N$ on single-particle wave packets to a bound that holds pointwise in $\vP$ on the number operator $N^b$ acting on the fiber eigenvectors $\Psi_{\vP,E_{\vP}}$.  In essence, this follows from the fundamental theorem of calculus. 
These arguments are implemented in Section \ref{Section-V}  and give the following results.\\

\begin{blanktheorem}[\textrm{\ref{theo-III.3}}]
For $|g|$ sufficiently small and $(\vP,E_{\vP}) \in I'_g\times \Delta_I$,
\begin{equation}\label{eq:III.36bis}
( \Psi_{\vP,E_{\vP}}\,,\,N_{n,\mathcal{C}_{\hat{P}}^2}^b\,\Psi_{\vP,E_{\vP}})\leq \cO(|g|^{\frac{(1-2\gamma)}{3}}|g|^{-\gamma/8}n^{\frac{4}{3}}\,\| |\vk|^{\beta}\,\one_{(\frac{1}{2(n+1)},\frac{3}{2n})}(\vk)\|_2^{\frac{1}{3}})\,,
\end{equation}
where
\begin{equation}
N_{n,\mathcal{C}_{\hat{P}}^2}^b:=d\Gamma^b(\chi_{n}^2(|\vk|)\,\xi_{\mathcal{C}_{\hat{P}}^2}^{g\,2}(\hat{k}))\,.
\end{equation}

Furthermore,  if in addition $||\vnabla E_\vP|-1|> \frac{3}{2}|g|^{\gamma/3}\,$ then
\begin{equation}\label{eq:III.38b-bis}
( \Psi_{\vP,E_{\vP}}\,,\,N_n^b\,\Psi_{\vP,E_{\vP}})\leq \cO(|g|^{\frac{(1-2\gamma)}{3}}n^{\frac{4}{3}}\| |\vk|^{\beta}\,\one_{(\frac{1}{2(n+1)},\frac{3}{2n})}(\vk)\|_2^{\frac{1}{3}}))\,
\end{equation}
where 
\begin{equation}
N_{n}^b:=d\Gamma^b(\chi_{n}^2(|\vk|))\,.
\end{equation}

\noindent
The constants in (\ref{eq:III.36bis}), (\ref{eq:III.38b-bis})  can be chosen uniformly in $\vec{P}$,  $|\vec{P|}\in I'_g \subset I$ ($I'_g$ is defined in Section \ref{subsec-II.2},  Eqs. (\ref{I-constraint}), (\ref{II-constraint})). They only depend on  $I$ and on $\Delta_I$.
\end{blanktheorem}

We now comment on the contents of Theorem \ref{theo-III.3}. Inequality (\ref{eq:III.36bis}) means that we can bound the boson number operator if we exclude a double cone (see Figure 3) and the definition of $\mathcal{C}_{\hat{P}}^2$ in Eqs. (\ref{eq:cone2-1})-(\ref{eq:cone2-3}), Section \ref{subsec-II.2}) around the direction of the  particle velocity, provided the form factor $\rho(\vk)$ scales like $|\vk|^{\beta}$ with $\beta >\betacrit$, i.e.\  \eqref{assumption on form factor}.

 The second result (see  (\ref{eq:III.38b-bis})) says that,  for putative mass shells $(\vP,E_{\vP})$ such that $||\vnabla E_\vP|-1|$ is not too small (i.e.,  $||\vnabla E_\vP|-1| > \frac{3}{2} |g|^{\gamma/3}$),  we can bound the boson number  without any angular restrictions, again using  that $\rho(\vk)$ scales like $|\vk|^{\beta}$ with $\beta >\betacrit$. The constraint means that  the forward emission of bosons by the (massive) particle  cannot be controlled if  its speed is too close to the boson propagation speed. 
\\
The estimates on the number operator obtained in Section \ref{Section-V} are used  in Section \ref{Section-VI}, where we will establish the following two results regarding the region $I_g\times \Delta_{I}$, where $I_g$ is any open interval contained in $I$ such that $|I_g|>|g|^{\gamma/2}$. \begin{itemize}
\item[(i)]
The first result is that we can exclude all the regular mass shells except those with slope close to $1$, i.e., all the \emph{regular} mass shells such that
\begin{equation}
(\vP,E_{\vP}) \in I_g\times \Delta_I \,\quad\text{and}\quad ||\vnabla E_\vP|-1|>\frac{3}{2}|g|^{\gamma/3}\,.
\end{equation}
\item[(ii)]
The second result shows that  a \emph{regular} mass shell might exist only for  $(\vP,E_{\vP})$ such that
\begin{equation}
 E_{\vP}=|\vP|-\frac{1}{2}+\cO(|g|^{\gamma/4})\,.
\end{equation}
\end{itemize}

\noindent
More precisely, we use that:
\begin{itemize}
\item[(1)]
The expectation value  in $\Psi_{\vP,E_{\vP}}$,  $|\vP| \in I'_g$, of  the operator $N^b$ restricted to the angular sector  $\mathcal{C}_{\hat{P}}^2$   vanishes  as $g\to0$ (see Theorem \ref{theo-III.3}).
\item[(2)]
For $(\vP,E_{\vP})\in I'_g\times \Delta_I$ such that $ ||\vnabla E_P|-1|>\cO(|g|^{\gamma/3})$,  the expectation value of the number operators $N^b$ on $\Psi_{\vP,E_{\vP}}$  vanishes as $g\to0$; see Theorem \ref{theo-III.3}. Analogously, if $ supp f_{\vec{Q}}^g \subset \{(\vP,E_{\vP})\in I'_g\times \Delta_I\,\,|\,\, ||\vnabla E_P|-1|>\cO(|g|^{\gamma/3})\}$ then  the expectation value of the number operator $N$ on $\Psi_{f_{\vec{Q}}}^g$  vanishes as $g\to0$; see Theorem \ref{theo-III.2}. 

\item[(3)]
The results in (2) imply that the putative fiber eigenvectors $\Psi_{\vP,E_{\vP}}$, $|\vP|\in I'_g$,  and the corresponding energies $E_{\vP}$ are pertubative in $g$ (see Corollary \ref{cor-V.6}),  provided that  $\beta>\betacrit$.
 \end{itemize}
 
 \noindent
We derive  (i) in Theorem \ref{theo-VI.1}  by mimicking  the argument with the trial states employed for  the treatment of the atomic resonances \cite{FP},  which was anticipated in Section  \ref{SectionIII.2}, Eqs (\ref{eq:trial})--(\ref{eq:root}). To this end,  we make us of (2) and (3). 

\noindent
The result in (ii) follows thanks to a stronger version of (1) (for details, see Lemma \ref{cor-III.6}, Lemma \ref{cor-III.7}) where only the  forward cone around the direction of the particle velocity is excluded in the definition of the restricted number operator, and  by combining  the eigenvalue equation with a standard (i.e., not a multi-scale analysis) virial argument in the fiber space $\cH_{\vP}$,  where the conjugate operator is $D_{\frac{1}{\kappa}, \kappa}^b\,:=\,d\Gamma^b(d_{\frac{1}{\kappa}, \kappa})$; see  (\ref{eq:III.7}) and Theorem \ref{theo-VI.4}.
\\

The virial identity exploited in Theorem \ref{theo-VI.4} is actually enough to exclude that, fiber by fiber, the eigenvalue lies at a distance larger than $\cO(g)$ above the unperturbed eigenvalue.  This observation is explained in the Remark after Theorem V.4 in Section V. However, the instability of the unperturbed mass shell proven in this paper requires a detailed analysis of the configuration of bosons in the putative eigenvector whose momenta are contained in different cones of momentum space. The decay mechanism exploited in Theorem \ref{theo-VI.1} combined with the assumed continuity  of the mass shell is responsible for the absence of  single-particle states except for the region $ E_{\vP}=|\vP|-\frac{1}{2}+\cO(|g|^{\gamma/4})$. This is because if the particle propagated at the critical velocity, i.e., $|\vnabla E_{P}|=1$, then there would be no kinematical constraint  preventing the emission of an arbitrarily large number of soft bosons in the forward direction (the direction of $\vP$).
\\


\section{Boson number estimates }\label{Section-V}
\resetequ

The main results in this section are Theorem \ref{theo-III.2}, Theorem \ref{theo-III.3}, and Corollary \ref{cor-V.6}.   Two preparatory results, contained in Section \ref{sec: preparatory lemmas}, are needed.  In particular, in Lemma \ref{lem-III.2},  we provide a rigorous justification of a virial identity employed in Lemma \ref{lem-V.3} and in Theorem \ref{theo-III.2}.

Since the proof of  Theorem \ref{theo-III.2} is lengthy, we present it in two different smaller sections: (a)  In Section \ref{Section-V.I.1},  we outline the proof of the theorem and, in Lemma  \ref{lem-V.3}, we introduce an important ingredient used later on. (b) In Section   \ref{Section-V.I.2}, we complete the steps of the  proof by assuming the result obtained in Lemma \ref{lem-V.3}.

In Section \ref{Section-V.2}, by using the regularity properties that follow from the \emph{Main Hypothesis}, we derive some estimates for the number operator $N^b$ evaluated on the fiber eigenvectors $\Psi_{\vP,E_{\vP}}$ analogous to those obtained in Theorem \ref{theo-III.2} for the number operator $N$ evaluated on the single-particle states $\Psi_{f^g_{\vec{Q}}}$. In Corollary  \ref{cor-V.6}, we then finally show that $E_{\vP}$ and $\Psi_{\vP,E_{\vP}}$ are perturbative in $g$,   provided $|\vP|\in I'_g$, $E_{\vP}\in \Delta_I$, $\beta>\betacrit$, and $||\vnabla E_{\vP}|-1|>\frac{3}{2} |g|^{\gamma/3}$.

\subsection{Preparatory results on virial identities} \label{sec: preparatory lemmas}
The following two lemmas are repeated and proven in Sections  \ref{appsec: lemma x domain} and \ref{appsec: lemma virial} of the appendix, respectively.
\begin{lemma} \label{lem-III.1}
The vector $ \Psi_{f_{\vec{Q}}^g}$ belongs to the domain of the position operator $\vx$ and 
\begin{equation}\label{eq:V.2}
 \norm x_j \Psi_{f_{\vec{Q}}^g}\| \leq \cO(|g|^{-\gamma} \|\Psi_{f_{\vec{Q}}^g}\|), \qquad  j=1,2,3\,.
\end{equation}
\end{lemma}
\begin{proof} See the appendix
\end{proof}


\noindent Lemma \ref{lem-III.2} states a virial theorem for our model.
We observe that by \emph{formal} steps one can derive the identity
\begin{eqnarray}\label{eq:III.9}
\,i[H-E_{\vP}\,,\,D_{n}^{\hat{u}}]
&= &\,d\Gamma(i[|\vk|,\,d_{n}^{\hat{u}}])-\,\vnabla E_{\vP}\cdot d\Gamma(i[\vk,\,d_{n}^{\hat{u}}])\\
& &-g\,[a^*(id_{n}^{\hat{u}} \rho_{\vx})+\, a(id_{n}^{\hat{u}} \rho_{\vx})]\,,\nonumber
\end{eqnarray}
where $E_{\vP}$, and $\vnabla E_{\vP}$ are operator-valued functions of the total momentum operator $\vP$.
Another formal step  would imply that 
\begin{equation}
0=(\Psi_{f_{\vec{Q}}^g}\,,\,i[H-E_{\vP}\,,\,D_{n}^{\hat{u}}]\,\Psi_{f_{\vec{Q}}^g})\,,
\end{equation}
and,  hence,
\begin{eqnarray}
0&= &(\Psi_{f_{\vec{Q}}^g}\,,\,d\Gamma(i[|\vk|,\,d_{n}^{\hat{u}}])\Psi_{f_{\vec{Q}}^g})-(\Psi_{f_{\vec{Q}}^g}\,,\,\vnabla E_{\vP}\cdot d\Gamma(i[\vk,\,d_{n}^{\hat{u}}])\Psi_{f_{\vec{Q}}^g}) \nonumber     \\
& &-g(\Psi_{f_{\vec{Q}}^g}\,,\,[a^*(id_{n}^{\hat{u}} \rho_{\vx})+\, a(id_{n}^{\hat{u}} \rho_{\vx})]\Psi_{f_{\vec{Q}}^g})\,. \label{eq:III.10} 
\end{eqnarray}
The next Lemma shows that all terms on the RHS of Eq. (\ref{eq:III.10}) can be given a well-defined meaning such that the  equality   is true.
\begin{lemma}\label{lem-III.2}
The identity
\begin{eqnarray}
0&= &(\Psi_{f_{\vec{Q}}^g}\,,\,d\Gamma(\chi_{n}^2(|\vk|)\xi^{g\,2}_{\hat{u}}(\hat{k})\,|\vk|)\Psi_{f_{\vec{Q}}^g})-(\Psi_{f_{\vec{Q}}^g}\,,\,\vnabla E_{\vP}\cdot d\Gamma(\chi_{n}^2(|\vk|)\xi^{g\,2}_{\hat{u}}(\hat{k})\,\vk)\Psi_{f_{\vec{Q}}^g})\nonumber\\
& &-g(\Psi_{f_{\vec{Q}}^g}\,,\,[a^*(id_{n}^{\hat{u}} \rho_{\vx})+\, a(id_{n}^{\hat{u}} \rho_{\vx})]\Psi_{f_{\vec{Q}}^g})\,\label{eq:III.11}
\end{eqnarray}
holds true. As the one-particle state $\Psi_{f_{\vec{Q}}^g}$ belongs to the form domain of all operators on the RHS of $\eqref{eq:III.11}$, this RHS is well-defined.
%
%
\end{lemma}
\begin{proof} See the appendix  
\end{proof}
Note that the formal equality of  the RHS of \eqref{eq:III.10}  and \eqref{eq:III.11} is straightforward. 
The virial identity of the Lemma above is first used in item (i) of \ref{Section-V.I.1}. A similar virial identity in item (ii) is proven analogously.

\subsection{Number operator estimates in putative single-particle states }\label{Section-V.I}

We now  proceed to proving the following theorem, where the expectation of the boson number operator in the state $ \Psi_{f_{\vec{Q}}^g}$ is bounded scale by scale.
\begin{theorem}\label{theo-III.2}
For $|g|$ sufficiently small,
\begin{equation}\label{eq:III.12a}
\frac{( \Psi_{f_{\vec{Q}}^g}\,,\,N_{n,\mathcal{C}_{\hat{Q}}^{1/2}}\,\Psi_{f_{\vec{Q}}^g})}{( \Psi_{f_{\vec{Q}}^g}\,,\,\Psi_{f_{\vec{Q}}^g})}\leq\, \cO(|g|^{2(1-2\gamma)}n^2 \| |\vk|^{\beta}\,\one_{(\frac{1}{2(n+1)},\frac{3}{2n})}(\vk)\|_2^2)\,,
\end{equation}
where $N_{n,\mathcal{C}_{\hat{Q}}^{1/2}}:=d\Gamma(\chi_{n}^2(|\vk|)\xi_{\mathcal{C}_{\hat{Q}}^{1/2}}^{g\,2}(\hat{k}))$ and $\vec{Q}\in I'_g$. 

Furthermore, if for all  $\vP\in supp f_{\vec{Q}}^g$  the inequality  $$||\vnabla E_{\vP}|-1|>|g|^{\gamma/3}$$ holds true  then
\begin{equation}\label{eq:III.14a}
\frac{( \Psi_{f_{\vec{Q}}^g}\,,\,N_{n}\,\Psi_{f_{\vec{Q}}^g})}{( \Psi_{f_{\vec{Q}}^g}\,,\,\Psi_{f_{\vec{Q}}^g})}\leq  \, \cO(|g|^{2(1-2\gamma)}n^2\| |\vk|^{\beta}\,\one_{(\frac{1}{2(n+1)},\frac{3}{2n})}(\vk)\|_2^2)\,,
\end{equation}
where $N_{n}:=d\Gamma(\chi_{n}^2(|\vk|)$. 

\noindent
The (implicit) constants in (\ref{eq:III.12a})-(\ref{eq:III.14a})  can be chosen to be uniform in $\vec{Q}$,  $|\vec{Q|}\in I'_g \subset I$; for the definition of $I'_g$ see Eqs. (\ref{I-constraint}), (\ref{II-constraint})).  They only depend on  $I$ and on $\Delta_I$.
\end{theorem}

\subsubsection{Outline of the proof of Theorem \ref{theo-III.2}}\label{Section-V.I.1}

\noindent
To prove inequalities (\ref{eq:III.12a}),  (\ref{eq:III.14a}) we exploit two different virial arguments and properties (P1), (P2), and (P3) of Sect. \ref{sec: properties}. More precisely,  we employ both conjugate operators $D_{n}^{\hat{u}}\,:=\,d\Gamma(d_{n}^{\hat{u}})$ and $D_{n,\perp}^{\hat{u},\hat{Q}}\,:=\,d\Gamma(d_{n,\perp}^{\hat{u},\hat{Q}})$, with $d_{n}^{\hat{u}}$ and $d_{n,\perp}^{\hat{u},\hat{Q}}$  defined in Eqs. (\ref{eq:II.12})  and (\ref{eq:II.10}), respectively. The  virial identities (see Lemma \ref{lem-III.2} for a rigorous treatment of the identities below) corresponding to $D_{n}^{\hat{u}}$ and $D_{n,\perp}^{\hat{u},\hat{Q}}$ are:
\begin{itemize}
\item[i)]
\begin{eqnarray}
\label{eq:III.12}
0&= &(\Psi_{f_{\vec{Q}}^g}\,,\,d\Gamma(i[|\vk|,\,d_{n}^{\hat{u}}])\Psi_{f_{\vec{Q}}^g})\\
& &-(\Psi_{f_{\vec{Q}}^g}\,,\,\vnabla E_{\vP}\cdot d\Gamma(i[\vk,\,d_{n}^{\hat{u}}])\Psi_{f_{\vec{Q}}^g})\\
& &-g(\Psi_{f_{\vec{Q}}^g}\,,\,[a^*(id_{n}^{\hat{u}} \rho_{\vx})+\, a(id_{n}^{\hat{u}} \rho_{\vx})]\Psi_{f_{\vec{Q}}^g})\,\\
&= &(\Psi_{f_{\vec{Q}}^g}\,,\,d\Gamma(\chi_{n}^2(|\vk|)\xi^{g\,2}_{\hat{u}}(\hat{k})\,|\vk|)\Psi_{f_{\vec{Q}}^g})\\
& &-(\Psi_{f_{\vec{Q}}^g}\,,\,\vnabla E_{\vP}\cdot d\Gamma(\chi_{n}^2(|\vk|)\xi^{g\,2}_{\hat{u}}(\hat{k})\,\vk)\Psi_{f_{\vec{Q}}^g})\label{eq:III.20a}\\
& &-g(\Psi_{f_{\vec{Q}}^g}\,,\,[a^*(id_{n}^{\hat{u}} \rho_{\vx})+\, a(id_{n}^{\hat{u}} \rho_{\vx})]\Psi_{f_{\vec{Q}}^g})\,;\label{eq:III.20c}
\end{eqnarray}
\item[ii)]
\begin{eqnarray}
\label{eq:V.13}
0&= &(\Psi_{f_{\vec{Q}}^g}\,,\,d\Gamma(i[|\vk|,\,d_{n,\perp}^{\hat{u},\hat{Q}}])\Psi_{f_{\vec{Q}}^g})\\
& &-(\Psi_{f_{\vec{Q}}^g}\,,\,\vnabla E_{\vP}\cdot d\Gamma(i[\vk,\, d_{n,\perp}^{\hat{u},\hat{Q}}])\Psi_{f_{\vec{Q}}^g})\\
& &-g(\Psi_{f_{\vec{Q}}^g}\,,\,[a^*(id_{n,\perp}^{\hat{u},\hat{Q}} \rho_{\vx})+\, a(id_{n,\perp}^{\hat{u},\hat{Q}} \rho_{\vx})]\Psi_{f_{\vec{Q}}^g})  \\
&=&(\Psi_{f_{\vec{Q}}^g}\,,\,d\Gamma( \chi_n^2(|\vk|) \xi^{g\,2}_{\hat{u}}(\hat{k})\,\frac{|\vk_{\perp}|^2}{|\vk|})\Psi_{f_{\vec{Q}}^g}) \label{eq: second virial 1} \\
& &-(\Psi_{f_{\vec{Q}}^g}\,,\,\vnabla E_{\vP}\cdot d\Gamma(\chi_{n}^2(|\vk|)\xi^{g\,2}_{\hat{u}}(\hat{k})\,\vk_{\perp})\Psi_{f_{\vec{Q}}^g}) \label{eq: second virial 2} \\
& &-g(\Psi_{f_{\vec{Q}}^g}\,,\,[a^*(id_{n,\perp}^{\hat{u},\hat{Q}} \rho_{\vx})+\, a(id_{n,\perp}^{\hat{u}, \hat{Q}} \rho_{\vx})]\Psi_{f_{\vec{Q}}^g}) \,.\label{eq: second virial 3}
\end{eqnarray}
\end{itemize}
(For the definition of the functions $\chi_{n}(|\vk|)$, $\xi^{g}_{\hat{u}}(\hat{k})$ see (a) and (b), in Section \ref{subsec-II.2}) 

\noindent
Next, we explain in detail the key role of the virial identities.
In order to arrive at inequalities (\ref{eq:III.12a}), (\ref{eq:III.14a}), we  study (see Lemma \ref{lem-V.3}) the number operator restricted to the sector associated with the unit vector $\hat{u}$, and derive  the estimate
\begin{equation}\label{eq:inequality}
( \Psi_{f_{\vec{Q}}^g}\,,\,N_{n,\hat{u}}\,\Psi_{f_{\vec{Q}}^g})\leq ( \Psi_{f_{\vec{Q}}^g}\,,\,\Psi_{f_{\vec{Q}}^g}) \, \cO(|g|^{2(1-\gamma-\tilde{\gamma})}n^2\||\vk|^{\beta}\,\one_{(\frac{1}{2(n+1)},\frac{3}{2n})}(\vk)\|_2^2)\,,
\end{equation}
where $N_{n,\hat{u}}:=d\Gamma(\chi_{n}^2(|\vk|)\xi_{\hat{u}}^{g\,2}(\hat{k}))$, for some $\tilde{\gamma}$,  $0<\tilde{\gamma}<\gamma$; we will eventually choose $\tilde{\gamma}=\gamma/2$.
In doing this, we start from the bound
\begin{eqnarray}\label{eq:bound}
& &|(\Psi_{f_{\vec{Q}}^g}\,,\,d\Gamma(\chi_{n}^2(|\vk|)\xi^{g\,2}_{\hat{u}}(\hat{k})\,|\vk|)\Psi_{f_{\vec{Q}}^g})-(\Psi_{f_{\vec{Q}}^g}\,,\,\vnabla E_{\vP}\cdot d\Gamma(\chi_{n}^2(|\vk|)\xi^{g\,2}_{\hat{u}}(\hat{k})\,\vk)\Psi_{f_{\vec{Q}}^g})| \nonumber \\
&& \qquad  \geq \cO(|g|^{\tilde{\gamma}})(\Psi_{f_{\vec{Q}}^g}\,,\,d\Gamma(\chi_{n}^2(|\vk|)\xi^{g\,2}_{\hat{u}}(\hat{k})\,|\vk|)\Psi_{f_{\vec{Q}}^g})\quad\quad\quad\quad\quad\quad\quad\quad\quad
\end{eqnarray}
that holds if, for all $\vP \in supp f_{\vec{Q}}^g$ and for all $\hat{k}$ in the sector,
\begin{equation}\label{eq:crucial}
|1-\hat{k}\cdot \vnabla E_{\vP}|>|g|^{\tilde{\gamma}}>0\,.
\end{equation}

\noindent
Given (\ref{eq:bound}), it is straightfoward to control the term (see (\ref{eq:III.20c})) associated with the interaction part of the Hamiltonian,  and to derive the inequality in (\ref{eq:inequality}). Therefore, the bound in  (\ref{eq:crucial}) is crucial, and we must identify the sectors where it is violated. We recall that $\vnabla E_{\vP}$ is collinear to $\vP$, and  we may assume that they are parallel; the other case can be treated in the same way. \\

First, note that the angle between $\vP \in supp f_{\vec{Q}}^g$ and $\vec{Q}$, as well as the angle between $\hat{u}$ and a vector $\vk$ that belongs to the sector associated with $\hat{u}$, are  $\cO(|g|^\gamma)$. This follows from the definitions of the function $ f_{\vec{Q}}^g$ and the cones $\mathcal{C}_{\hat u}$, given in Section \ref{subsec-II.2}. It implies that, roughly speaking, we can identify $\hat P=\hat{Q}$ and $\hat{k}=\hat u $, since, for $|g|$ small enough,  $|g|^\gamma$ is much smaller than  $|g|^{\tilde{\gamma}}$ in \eqref{eq:crucial}.  \\
The vectors $\vk$ for which \eqref{eq:crucial} fails, satisfy
\begin{equation} 
\left\str \hat{k}\cdot \hat{P}- | \vnabla E_{\vP}|^{-1} \right\str \leq \cO( |g|^{\tilde{\gamma}(=\gamma/2)}),\, \qquad  |\vnabla E_{\vP}|>0\,.
\end{equation}
Hence, if $| \vnabla E_{\vP}| $ is bounded away from $1$, either - for $|\vnabla E_{\vP}|<1$; see also {\bf{(B)}} in Section  \ref{Section-V.I.2} - the condition \eqref{eq:crucial} is always satisfied, or - for  $|\vnabla E_{\vP}|>1$; see also {\bf{(C)}} in Section  \ref{Section-V.I.2} - such $\vk$  have a nonvanishing component, $\vk_{\perp}$, (of order $\sqrt{1-| \vnabla E_{\vP}|^{-2}} $) in the orthogonal complement of $\vec{Q} (=\vec{P})$. 
In particular, they satisfy
\begin{equation}\label{eq: inequality for difficult sectors}
\big|\frac{|\vk_{\perp}|^2}{|\vk|^2}-\frac{|\vk_{\perp}|}{|\vk|}\hat{k}_{\perp}\cdot \vnabla E_{\vP}\big| > \cO(|g|^{\gamma/3})>0\,,
\end{equation}
 Note that the second term on the LHS of \eqref{eq: inequality for difficult sectors} actually vanishes if our approximation $\hat P=\hat Q$ were to hold exactly.
In   Section \ref{Section-V.I.2}, we  establish \eqref{eq: inequality for difficult sectors} rigorously. 
  The bound \eqref{eq: inequality for difficult sectors}  immediately implies that, for the sectors $\hat{u}$  for which \eqref{eq:crucial} fails,  the following bound holds true
\begin{eqnarray}
& &|(\Psi_{f_{\vec{Q}}^g}\,,\,d\Gamma(\chi_n^2(|\vk|) \xi^{g\,2}_{\hat{u}}(\hat{k})\,\frac{|\vk_{\perp}|^2}{|\vk|})\Psi_{f_{\vec{Q}}^g})\nonumber \\
& &\,\,-(\Psi_{f_{\vec{Q}}^g}\,,\,\vnabla E_{\vP}\cdot d\Gamma(\chi_{n}^2(|\vk|)\xi^{g\,2}_{\hat{u}}(\hat{k})\,\vk_{\perp})\Psi_{f_{\vec{Q}}^g})| \nonumber \\
&\geq  &\cO(|g|^{\gamma/3})(\Psi_{f_{\vec{Q}}^g}\,,\,d\Gamma(\chi_n^2(|\vk|) \xi^{g\,2}_{\hat{u}}(\hat{k})\,|\vk|)\Psi_{f_{\vec{Q}}^g})\,,
\end{eqnarray}
Starting from this bound, we can use the second virial identity  (\ref{eq: second virial 1}, \ref{eq: second virial 2}, \ref{eq: second virial 3})  to derive the inequality (\ref{eq:inequality})  for the sectors $\hat u$ for which  \eqref{eq:crucial} fails. 

The conclusion is that, under the condition that $|\vnabla E_{\vP}|$, $\vP\in supp f_{\vec{Q}}^g$, differs from $1$ by a quantity $>\cO(|g|^{\gamma/3})$, we can cover all the sectors by the two virial identities above. Without the restriction on $|\vnabla E_{\vP}|$, these arguments only show  that Eq. (\ref{eq:inequality}) holds for all $\hat{u}$-dependent sectors contained in the  cone $\mathcal{C}^{1/2}_{\hat{Q}}$.\\

In implementing this strategy, we make use of the following lemma.

\begin{lemma} \label{lem-V.3}
Fix a unit vector $\hat{u}$ and assume that,
 for all $\vP \in supp f_{\vec{Q}}^g$ and for all $\hat{k}\in supp\, \xi_{\hat{u}}^g$, 
\begin{equation}\label{eq:III.26-bis}
|1-\hat{k}\cdot \vnabla E_{\vP}|>|g|^{\tilde{\gamma}}>0\,,\quad\quad 0<\tilde{\gamma}<\gamma\,,
\end{equation}
where $ \tilde{\gamma}$ is $g$- and $\vec{Q}$-independent.
Then, for $|g|$ small enough,  the following bound holds true
\begin{equation}\label{eq:V.16}
( \Psi_{f_{\vec{Q}}^g}\,,\,N_{n,\hat{u}}\,\Psi_{f_{\vec{Q}}^g})\leq ( \Psi_{f_{\vec{Q}}^g}\,,\,\Psi_{f_{\vec{Q}}^g}) \, \cO(|g|^{2(1-\gamma-\tilde{\gamma})}n^2\| |\vk|^{\beta}\,\one_{(\frac{1}{2(n+1)},\frac{3}{2n})}(\vk)\|_2^2)\,,
\end{equation}
where $N_{n,\hat{u}}:=d\Gamma(\chi_{n}^2(|\vk|)\xi_{\hat{u}}^{g\,2}(\hat{k}))$.
\end{lemma}

\noindent
\emph{Proof}

\noindent
We assume that (\ref{eq:III.26-bis}) holds with $$1-\hat{k}\cdot \vnabla E_{\vP}<0;$$ the other case, $1-\hat{k}\cdot \vnabla E_{\vP}>0$, can be treated  similarly.  We get
\begin{eqnarray}
0&=&(\Psi_{f_{\vec{Q}}^g}\,,\,d\Gamma(\chi_{n}^2(|\vk|)\xi_{\hat{u}}^{g\,2}(\hat{k})\,|\vk|)\Psi_{f_{\vec{Q}}^g})\quad\,\, \label{eq:III.43bis}\\
& &-(\Psi_{f_{\vec{Q}}^g}\,,\,\vnabla E_{\vP}\cdot d\Gamma(\chi_{n}^2(|\vk|)\xi_{\hat{u}}^{g\,2}(\hat{k})\,\vk)\Psi_{f_{\vec{Q}}^g})\quad\quad\quad\,\,  \nonumber \\
& &-g(\Psi_{f_{\vec{Q}}^g}\,,\,[a^*(id_{n}^{\hat{u}} \rho_{\vx})+\, a(id_{n}^{\hat{u}} \rho_{\vx})]\Psi_{f_{\vec{Q}}^g})\,\nonumber \\
& \leq &-|g|^{\tilde{\gamma}}(\Psi_{f_{\vec{Q}}^g}\,,\,d\Gamma(\chi_{n}^2(|\vk|)\xi_{\hat{u}}^{g\,2}(\hat{k})|\vk|)\Psi_{f_{\vec{Q}}^g})\label{eq:III.43bisbis} \\
& &
+c|g|^{1-\gamma}\|\Psi_{f_{\vec{Q}}^g}\|\|d\Gamma(\chi_{n}^2(|\vk|)\xi_{\hat{u}}^{g\,2}(\hat{k}))^{1/2}\Psi_{f_{\vec{Q}}^g} \|\times\\
& &\quad \times(\int  |\vk|^{2\beta}\,\one_{(\frac{1}{2(n+1)},\frac{3}{2n})}(\vk)\,d^3k)^{1/2}\nonumber 
\end{eqnarray}
for some constant $c$, $c>0$, uniform in $\vec{Q}$,  $|\vec{Q}|\in I'_g$. To do the step from (\ref{eq:III.43bis}) to (\ref{eq:III.43bisbis}), we split
\begin{eqnarray}
i(d_{n}^{\hat{u}} \rho_{\vx})(\vk)
&= &-\,\chi_{n}(|\vk|)\,\xi_{\hat{u}}^g(\hat{k})\,\big(\vk\cdot\vnabla_{\vk}\big(\chi_n(|\vk|)\xi_{\hat{u}}^g(\hat{k})\big) \big)\, {\rho}(\vk)\,e^{-i  \vk \cdot \vx} \,\quad\quad\quad    \nonumber\\
& &-\,\chi_{n}^2(|\vk|)\xi_{\hat{u}}^{g\,2}(\hat{k})\,\big(\vk\cdot\vnabla_{\vk}
{\rho}(\vk)\big) e^{-i  \vk \cdot \vx}   \nonumber\\
& &-\frac{1}{2}\,\chi_{n}^2(|\vk|)\xi_{\hat{u}}^{g\,2}(\hat{k})\,\big(\vnabla_{\vk}\cdot\vk\big)\,{\rho}(\vk)\; e^{-i \vk \cdot \vx}  \nonumber\\
& &-\,\chi_{n}^2(|\vk|)\xi_{\hat{u}}^{g\,2}(\hat{k})\, {\rho}(\vk)\, \big(\vk\cdot\vnabla_{\vk}e^{-i  \vk \cdot \vx}\big)\,  \label{eq: splitting of form factor}
\end{eqnarray} 
and we may justify this step for each of the four terms separately, using the Schwarz inequality and
\begin{itemize}
\item[i)] The assumption $\left\str 1-\hat{k}\cdot \vnabla E_{\vP}\right\str>|g|^{\tilde{\gamma}}$ for  $\vP \in supp f_{\vec{Q}}^g$;
\item[ii)] The infrared behavior of $\rho(\vk)$ as assumed in \eqref{assumption on form factor}, i.e., $|\rho(\vk)| \leq\cO( |\vk|^{\beta})$ and $|\vnabla_{\vk} {\rho}(\vk)| \leq\cO( |\vk|^{\beta-1})$;
\item[iii)] Lemma \ref{lem-III.1}.
\end{itemize}
As an example, for the term proportional to $$ -\,\chi_{n}^2(|\vk|)\xi_{\hat{u}}^{g\,2}(\hat{k})\, {\rho}(\vk)\, \big(\vk\cdot\vnabla_{\vk}e^{-i  \vk \cdot \vx}\big)$$ we proceed as follows:
\begin{eqnarray}
\lefteqn{ \left|(\Psi_{f_{\vec{Q}}^g}\,,\,a^*\big(-\,\chi_{n}^2(|\vk|)\xi_{\hat{u}}^{g\,2}(\hat{k})\, \rho(\vk)\, \big(\vk\cdot\vnabla_{\vk}e^{-i  \vk \cdot \vx}\big)\big)\Psi_{f_{\vec{Q}}^g}) \right| } \nonumber \\
&= &\left|\int \,\chi_{n}^2(|\vk|)\xi_{\hat{u}}^{g\,2}(\hat{k})\, \rho(\vk)\, \vk\cdot(a_{\vk}\Psi_{f_{\vec{Q}}^g}\,,\,\vnabla_{\vk}e^{-i  \vk \cdot \vx}\Psi_{f_{\vec{Q}}^g})d^3k \right| \nonumber 
\\
&\leq& \int \,\chi_{n}^2(|\vk|)\xi_{\hat{u}}^{g\,2}(\hat{k})\, |\rho(\vk)| \, |\vk| \, |(a_{\vk}\Psi_{f_{\vec{Q}}^g}\,,\,\vnabla_{\vk}e^{-i  \vk \cdot \vx}\Psi_{f_{\vec{Q}}^g})|  d^3k\nonumber \\
&\leq&\int \,\chi_{n}^2(|\vk|)\xi_{\hat{u}}^{g\,2}(\hat{k})\, |\rho(\vk)|\, |\vk| \, \|a_{\vk}\Psi_{f_{\vec{Q}}^g}\| \|\vnabla_{\vk}e^{-i  \vk \cdot \vx}\Psi_{f_{\vec{Q}}^g}\|d^3k \nonumber \\
&\leq&\big(\int \,\chi_{n}^2(|\vk|)\xi_{\hat{u}}^{g\,2}(\hat{k})\,\|a_{\vk}\Psi_{f_{\vec{Q}}^g}\|^2 d^3k\big)^{1/2}\times \nonumber \\
& &\times\big(\int\|\vnabla_{\vk}e^{-i \vk \cdot \vx}\Psi_{f_{\vec{Q}}^g}\|^2\, |\vk|^{2\beta}\,\one_{(\frac{1}{2(n+1)},\frac{3}{2n})}(\vk)\, |\vk|^2d^3k\big)^{1/2}\,.
\end{eqnarray} 
We notice that
\begin{equation}
\big(\int \,\chi_{n}^2(|\vk|)\xi_{\hat{u}}^{g\,2}(\hat{k})\,\|a_{\vk}\Psi_{f_{\vec{Q}}^g}\|^2 d^3k\big)^{1/2}=\|d\Gamma(\chi_{n}^2(|\vk|)\xi_{\hat{u}}^{g\,2}(\hat{k}))^{1/2}\Psi_{f_{\vec{Q}}^g} \|
\end{equation}
and, since $\|\vnabla_{\vk}e^{-i  \vk \cdot \vx}\Psi_{f_{\vec{Q}}^g}\|\leq \cO(|g|^{-\gamma}\|\Psi_{f_{\vec{Q}}^g}\|)$ (Lemma \ref{lem-III.1}),
\begin{eqnarray}
\lefteqn{\big(\int\|\vnabla_{\vk}e^{-i  \vk \cdot \vx}\Psi_{f_{\vec{Q}}^g}\|^2|\vk|^{2\beta}\,\one_{(\frac{1}{2(n+1)},\frac{3}{2n})}(\vk)\, |\vk|^2d^3k\big)^{1/2}}\\
& \leq& C |g|^{-\gamma} (\int  |\vk|^{2\beta}\,\one_{(\frac{1}{2(n+1)},\frac{3}{2n})}(\vk)\,d^3k)^{1/2}\|\Psi_{f_{\vec{Q}}^g}\|\,.
\end{eqnarray}

\noindent
Then, starting from (\ref{eq:III.43bisbis}), the bound from above takes the form
\begin{equation}
( \Psi_{f_{\vec{Q}}^g}\,,\,N_{n,\hat{u}}\,\Psi_{f_{\vec{Q}}^g})\leq ( \Psi_{f_{\vec{Q}}^g}\,,\,\Psi_{f_{\vec{Q}}^g}) \, \cO(|g|^{2(1-\gamma-\tilde{\gamma})}n^2\||\vk|^{\beta}\,\one_{(\frac{1}{2(n+1)},\frac{3}{2n})}(\vk)\|_2^2)\,,
\end{equation}
where $N_{n,\hat{u}}:=d\Gamma(\chi_{n}^2(|\vk|)\xi_{\hat{u}}^{g\,2}(\hat{k}))$, because
\begin{equation}
\chi_{n}^2(|\vk|)\xi_{\hat{u}}^{g\,2}(\hat{k})|\vk| \geq \frac{1}{2(n+1)}\chi_{n}^2(|\vk|)\xi_{\hat{u}}^{g\,2}(\hat{k})\,.
\end{equation}
\QED
\subsubsection{Proof of Theorem \ref{theo-III.2}} \label{Section-V.I.2}
Notice that, starting from Lemma  \ref{lem-V.3}, we can fill the region 
\begin{equation}
\{\hat{k}:|1-\hat{k}\cdot \vnabla E_{\vP}|>|g^{\tilde{\gamma}}|\quad\forall \vP\in supp f_{\vec{Q}}^g\}
\end{equation}
with sectors corresponding to functions $\xi_{\hat{u}_j}^g$ where  $1\leq j \leq \bar{j}\leq  \cO(|g|^{-\gamma})$, so that we obtain
\begin{equation}\label{eq:III.30}
\sum_{j=1}^{\bar{j}}( \Psi_{f_{\vec{Q}}^g}\,,\,N_{n,\hat{u}_j}\,\Psi_{f_{\vec{Q}}^g})\leq ( \Psi_{f_{\vec{Q}}^g}\,,\,\Psi_{f_{\vec{Q}}^g}) \cO(|g|^{2(1-\frac{3\gamma}{2}-\tilde{\gamma})}n^2\||\vk|^{\beta}\,\one_{(\frac{1}{2(n+1)},\frac{3}{2n})}(\vk)\|_2^2)\,.
\end{equation}

We observe that if,  for some $\vP\in  supp f_{\vec{Q}}^g$,  $$||\vnabla E_{\vP}|-1|\leq |g|^{\gamma/3}\,$$then, for $|g|$ small enough,   $$||\vnabla E_{\vP'}|-1|\leq2 |g|^{\gamma/3}$$  for all $\vP'\in  supp f^g_{\vec{Q}}$.
This holds because
\begin{itemize}
\item
of the constraints on the support of $f_{\vec{Q}}^g$ (see Section \ref{subsec-II.2});
\item
$|\frac{\partial^2 E_{\vP}}{\partial |\vP|^2}|\leq C''_{I}$; (see Property  (P2) in Section \ref{subsec-II.1}).
\end{itemize}
After the  result in Eq. (\ref{eq:III.30}),  which holds for sectors  such that (\ref{eq:III.26-bis}) (Lemma (\ref{lem-V.3})) is fulfilled,  we may distinguish three possible situations,  {\bf{(A)}}, {\bf{(B)}},  and {\bf{(C)}}, depending on the length of the vector $\vnabla E_{\vP}$, $\vP \in supp f_{\vec{Q}}^g$.

\begin{itemize}
\item[{\bf{(A)}}]

\noindent
\emph{For some $\vP \in supp f_{\vec{Q}}^g$,  $||\vnabla E_{\vP}|-1|\leq |g|^{\gamma/3}$.}

\noindent
In this case, $\forall \hat{k} \in \mathcal{C}_{\hat{Q}}^{1/2}$,  the inequality in  (\ref{eq:III.26-bis})  holds true for $\tilde{\gamma}=\gamma/2$ and  $\str g \str$ small enough, because
\begin{itemize}
\item[i)]
$||\vnabla E_{\vP'}|-1|\leq 2 |g|^{\gamma/3}$, $\forall \,\vP' \in \, supp f^g_{\vec{Q}}$.
\item[ii)]
by definition $$\mathcal{C}_{\hat{Q}}^{1/2}:=\{\hat{k}\,\,:\,\,|\hat{k}\cdot\hat{Q}| \leq \cos(\frac{1}{2}|g|^{\gamma/8})\}\,.$$
\end{itemize}
Thus, we can use the estimate in Eq.  (\ref{eq:III.30}) with $\tilde{\gamma}=\gamma/2$, 
\begin{equation}
\frac{( \Psi_{f_{\vec{Q}}^g}\,,\,N_{n,\mathcal{C}_{\hat{Q}}^{1/2}}\,\Psi_{f_{\vec{Q}}^g})}{( \Psi_{f_{\vec{Q}}^g}\,,\,\Psi_{f_{\vec{Q}}^g})}\leq\, \cO(|g|^{2(1-2\gamma)}n^2\| |\vk|^{\beta}\,\one_{(\frac{1}{2(n+1)},\frac{3}{2n})}(\vk)\|_2^2)\,,
\end{equation}
where $N_{n,\mathcal{C}_{\hat{Q}}^{1/2}}:=d\Gamma(\chi_{n}^2(|\vk|)\xi_{\mathcal{C}_{\hat{Q}}^{1/2}}^{g\,2}(\hat{k}))$ ($\chi_{n}(|\vk|)$ and $\xi_{\mathcal{C}_{\hat{Q}}^{1/2}}^g(\hat{k})$ are defined in (a) and (b) of Section (\ref{subsec-II.2})).
\item[{\bf{(B)}}]

\noindent
\emph{For some $\vP \in supp f_{\vec{Q}}^g$,  $|\vnabla E_{\vP}|<1-|g|^{\gamma/3}$.}

\noindent
The constraint (\ref{eq:III.26-bis}) with $\tilde{\gamma}=\gamma/2$ is  fulfilled  for all angular sectors.
\item[{\bf{(C)}}]

\noindent
\emph{For all $\vP\in supp f_{\vec{Q}}^g$
$,  |\vnabla E_{\vP}|> 1+|g|^{\gamma/3}$.}

\noindent
First we notice that  we can restrict our analysis to an angular sector labeled by a direction $\hat{u}$ such that, for some $\vP\in supp f_{\vec{Q}}^g$, the inequality
\begin{equation}\label{eq:III.31b}
|1-\hat{k}\cdot \vnabla E_{\vP}|\leq |g|^{\tilde{\gamma}(=\gamma/2)}
\end{equation}
holds true for some $\hat{k}$ belonging to the sector under consideration. This is because, if \begin{equation}\label{eq:III.31c}
|1-\hat{k}\cdot \vnabla E_{\vP}|>|g|^{\tilde{\gamma}(=\gamma/2)}\,,
\end{equation}
for all $\hat{k}$ belonging to the given sector,
 then the result in (\ref{eq:V.16}) holds,  as we have proven in Lemma \ref{lem-V.3}.  

\noindent
We now show that the combination of $|\vnabla E_{\vP}|> 1+|g|^{\gamma/3}$ and (\ref{eq:III.31b}) yields the useful inequality (\ref{eq:III.42}) below. 

\noindent
We notice that, assuming the bound in Eq. (\ref{eq:III.31b}) for some $\hat{k}$ belonging to the sector,  for $|g|$ small enough,
 \begin{equation}
|1-\hat{k}\cdot \vnabla E_{\vP}|\leq2 |g|^{\tilde{\gamma}(=\gamma/2)}\,, \label{eq: gradient close to one}
\end{equation}
for all $\hat{k}$ in the sector labeled by $\hat{u}$.
\noindent
Furthermore, $\hat{P}\cdot \hat{Q}\geq \cos(|g|^{\gamma})$, for all $\vP \in supp f_{\vec{Q}}^g$, by construction, and we may assume that $\vnabla E_{\vP}$ is parallel to $\vP$; the other case, $\vP\cdot \vnabla E_{\vP}=-|\vP| | \vnabla E_{\vP}|$,  can be treated in an analogous way. Let $\eta$ be  the angle between $\hat{k}$ and $\hat{Q}$. Then,  (\ref{eq: gradient close to one}) means  that
\begin{equation}
-2|g|^{\gamma/2}\leq 1 -\cos(\eta+\epsilon)|\vnabla E_{\vP}|\leq 2|g|^{\gamma/2}\,,
\end{equation}
where $\epsilon=\cO(|g|^{\gamma})$ and, for $|g|$ small enough,  
\begin{equation}
1-c'|g|^{\gamma/3}\geq\frac{[1+2|g|^{\gamma/2}]}{|\vnabla E_{\vP}|}\geq \cos(\eta+\epsilon)\geq \frac{[1-2|g|^{\gamma/2}]}{|\vnabla E_{\vP}|}
\end{equation}
for some constant $c'>0$. Hence we have $\eta \geq c''|g|^{\gamma/6}>0$ where $c''>0$, and we find that
\begin{equation}\label{eq:III.42}
\big|\frac{|\vk_{\perp}|^2}{|\vk|^2}-\frac{|\vk_{\perp}|}{|\vk|}\hat{k}_{\perp}\cdot \vnabla E_{\vP}\big|>\cO(\sin^2(\eta))\geq \cO(|g|^{\gamma/3})>0\,,
\end{equation}
for all $\vk$ in the sector, where  $\vk_{\perp}:=\vk-\frac{\vk\cdot\vec{Q}}{|\vec{Q}|^2}\,\vec{Q}$, because
\begin{itemize}
\item[i)]
by assumption, $\hat{P}\cdot \hat{Q}>\cos(|g|^{\gamma})$; 
\item[ii)]
$\vnabla E_{\vP}$ is parallel (or antiparallel) to $\vP$ and and $|\vnabla E_{\vP}|<C'_{I}$;
\item[iii)]
$\frac{|\vk_{\perp}|}{|\vk|}=\sin(\eta)$;
 \item[iv)]
$|\hat{k}_{\perp}\cdot \vnabla E_{\vP}|\leq|\vnabla E_{\vP}|\times \cO(|g|^{\gamma})$ using that $\hat{P}\cdot \hat{Q}>\cos(|g|^{\gamma})$.
\end{itemize}

\noindent
Assuming, for example, that (\ref{eq:III.42}) holds, because
\begin{equation} \label{eq:III.32}
\frac{|\vk_{\perp}|^2}{|\vk|^2}-\frac{|\vk_{\perp}|}{|\vk|}\hat{k}_{\perp}\cdot \vnabla E_{\vP}<-c|g|^{\gamma/3}\,,
\end{equation}
where $c>0$, we use the second virial identity (see Eq. (\ref{eq:V.13})) to obtain 
\begin{eqnarray}
0&=&(\Psi_{f_{\vec{Q}}^g}\,,\,d\Gamma(\chi_{n}^2(|\vk|)\xi_{\hat{u}}^{g\,2}(\hat{k})\,\frac{|\vk_{\perp}|^2}{|\vk|})\Psi_{f_{\vec{Q}}^g})\\
& &-(\Psi_{f_{\vec{Q}}^g}\,,\,\vnabla_{\vP}E_{\vP}\cdot d\Gamma(\chi_{n}^2(|\vk|)\xi_{\hat{u}}^{g\,2}(\hat{k})\,\vk_{\perp})\Psi_{f_{\vec{Q}}^g})\nonumber \\
& &-g(\Psi_{f_{\vec{Q}}^g}\,,\,[a^*(id_{n,\perp}^{\hat{u},\hat{P'}} \rho_{\vx})+\, a(id_{n,\perp}^{\hat{u},\hat{P'}} \rho_{\vx})]\Psi_{f_{\vec{Q}}^g})\,,\nonumber \\
& \leq &-c|g|^{\gamma/3}(\Psi_{f_{\vec{Q}}^g}\,,\,d\Gamma(\chi_{n}^2(|\vk|)\xi_{\hat{u}}^{g\,2}(\hat{k})|\vk|)\Psi_{f_{\vec{Q}}^g})\label{eq:V.67b} \\
& &
+c'|g|^{1-\gamma}\|\Psi_{f_{\vec{Q}}^g}\|\|d\Gamma(\chi_{n}^2(|\vk|)\xi_{\hat{u}}^{g\,2}(\hat{k}))^{1/2}\Psi_{f_{\vec{Q}}^g} \|\times\\
& &\quad \times(\int  |\vk|^{2\beta}\,\one_{(\frac{1}{2(n+1)},\frac{3}{2n})}(\vk)\,d^3k)^{1/2} \nonumber
\end{eqnarray}
for some $c'>0$. Now we estimate (\ref{eq:V.67b}) similarly to (\ref{eq:III.43bisbis}) in Lemma \ref{lem-V.3}.
\end{itemize}
\noindent
\emph{Conclusions}

\noindent
For $|g|$ small enough, we have proven that:
\begin{itemize}
\item[i)] By combining cases {\bf{(A)}}, {\bf{(B)}}, and {\bf{(C)}}, 
 \begin{equation}
\frac{( \Psi_{f_{\vec{Q}}^g}\,,\,N_{n,\mathcal{C}_{\hat{Q}}^{1/2}}\,\Psi_{f_{\vec{Q}}^g})}{( \Psi_{f_{\vec{Q}}^g}\,,\,\Psi_{f_{\vec{Q}}^g})}\leq\, \cO(|g|^{2(1-2\gamma)}n^2\| |\vk|^{\beta}\,\one_{(\frac{1}{2(n+1)},\frac{3}{2n})}(\vk)\|_2^2)\,.
\end{equation}
Note that, in cases {\bf{(B)}} and {\bf{(C)}}, the angular restriction was not used.
\item [ii)] Under the assumption that $$||\vnabla E_{\vP}|-1|>|g|^{\gamma/3},\quad \text{for all}\quad \vP\in f^g_{\vec{Q}},$$
we have that
\begin{equation}
\frac{( \Psi_{f_{\vec{Q}}^g}\,N_{n}\,\Psi_{f_{\vec{Q}}^g})}{( \Psi_{f_{\vec{Q}}^g},\Psi_{f_{\vec{Q}}^g})}\leq  \, \cO(|g|^{2(1-2\gamma)}n^2\| |\vk|^{\beta}\,\one_{(\frac{1}{2(n+1)},\frac{3}{2n})}(\vk)\|_2^2)\,,
\end{equation}
where $N_{n}:=d\Gamma(\chi_{n}^2(|\vk|)$. This follows from cases {\bf{(B)}} and {\bf{(C)}}.
\end{itemize}
\QED

\subsection{Number operator estimates in putative fiber eigenvectors }\label{Section-V.2}

Using the results in Theorem \ref{theo-III.2} and Property (P3), we are now in a position to state some bounds on the expectation value of the boson number operator restricted to the  fiber spaces.  These bounds hold pointwise in $\vP$,  for $|\vP|$ in  the open interval  $I'_g\subset I_g$ introduced in Section  \ref{subsec-II.2}; see Eqs. (\ref{I-constraint}, \ref{II-constraint}).

\begin{theorem}\label{theo-III.3}
For $|g|$ sufficiently small, and $(\vP,E_{\vP}) \in I'_g\times \Delta_I$:
\begin{equation}\label{eq:III.36}
( \Psi_{\vP,E_{\vP}}\,,\,N_{n,\mathcal{C}_{\hat{P}}^2}^b\,\Psi_{\vP,E_{\vP}})\leq \cO(|g|^{\frac{(1-2\gamma)}{3}}|g|^{-\gamma/8}n^{\frac{4}{3}}\| |\vk|^{\beta}\,\one_{(\frac{1}{2(n+1)},\frac{3}{2n})}(\vk)\|_2^{\frac{1}{3}})\,,
\end{equation}
where
\begin{equation}
N_{n,\mathcal{C}_{\hat{P}}^2}^b:=d\Gamma^b(\chi_{n}^2(|\vk|)\,\xi_{\mathcal{C}_{\hat{P}}^{2}}^{g \,2}(\hat{k}))\,.
\end{equation}
\\

\noindent
Furthermore,  if  in addition  $||\vnabla E_\vP|-1|> \frac{3}{2}|g|^{\gamma/3}\,$  then
\begin{equation}\label{eq:III.38b}
( \Psi_{\vP,E_{\vP}}\,,\,N_n^b\,\Psi_{\vP,E_{\vP}})\leq \cO(|g|^{\frac{(1-2\gamma)}{3}}n^{\frac{4}{3}}\| |\vk|^{\beta}\,\one_{(\frac{1}{2(n+1)},\frac{3}{2n})}(\vk)\|_2^{\frac{1}{3}})\,
\end{equation}
where 
\begin{equation}
N_{n}^b:=d\Gamma^b(\chi_{n}^2(|\vk|))\,.
\end{equation}
\noindent
The constants in (\ref{eq:III.36}), (\ref{eq:III.38b})  can be chosen uniformly in $\vec{P}$,  $|\vec{P|}\in I'_g \subset I$ ($I'_g$ is defined in Section \ref{subsec-II.2},  Eqs. (\ref{I-constraint}), (\ref{II-constraint})). They only depend on  $I$ and on $\Delta_I$.
\end{theorem}

\noindent
\emph{Proof}

\noindent
First of all, we observe that, for  $\vP$ such that $f_{\vec{Q}}^g(\vP)=1$,  the inequality
\begin{equation}
( \Psi_{\vP,E_{\vP}}\,,\,N_{n,\mathcal{C}_{\hat{Q}}^{1/2}}^b\,\Psi_{\vP,E_{\vP}}) \leq \cO(|g|^{(1-2\gamma)}n\| |\vk|^{\beta}\,\one_{(\frac{1}{2(n+1)},\frac{3}{2n})}(\vk)\|_2)\quad\
\end{equation}
can fail to hold true only for $\vP$ in a set $I^{*}_{f_{\vec{Q}}^g}$ of measure  bounded above by \begin{equation}
( \Psi_{f_{\vec{Q}}^g},\Psi_{f_{\vec{Q}}^g})\,\cO(g^{(1-2\gamma)}n\| |\vk|^{\beta}\,\one_{(\frac{1}{2(n+1)},\frac{3}{2n})}(\vk)\|_2)\,,
\end{equation}
 i.e., 
\begin{equation}\label{eq:III.41}
\int_{I^{*}_{f_{\vec{Q}}^g}}d^3P \leq ( \Psi_{f_{\vec{Q}}^g},\Psi_{f_{\vec{Q}}^g})\, \cO(|g|^{(1-2\gamma)}n\| |\vk|^{\beta}\,\one_{(\frac{1}{2(n+1)},\frac{3}{2n})}(\vk)\|_2)\,.
\end{equation}
This follows from inequality (\ref{eq:III.12a}),  which we can write as
\begin{eqnarray}
\lefteqn{\int d^3P\,\bar{f}^g_{\vec{Q}}(\vP)f_{\vec{Q}}^g(\vP)( \Psi_{\vP,E_{\vP}}\,,\,N_{n,\mathcal{C}_{\hat{Q}}^{1/2}}^b\, \Psi_{\vP,E_{\vP}})}\\
& &\leq ( \Psi_{f_{\vec{Q}}^g}\,,\,\Psi_{f_{\vec{Q}}^g}) \, \cO(|g|^{2(1-2\gamma)}n^2\| |\vk|^{\beta}\,\one_{(\frac{1}{2(n+1)},\frac{3}{2n})}(\vk)\|_2^2)\,.
\end{eqnarray}
Next,  we make use of  the following inequality, which holds  in the sense of quadratic forms, 
\begin{equation}\label{eq:V.77bis}
N_{n,\mathcal{C}_{\hat{P}}^2}^b\,\leq\,N_{n,\mathcal{C}_{\hat{Q}}^{1/2}}^b
\end{equation}
for $\vP$ in the support of $f_{\vec{Q}}^g$. This inequality can be easily derived from the definitions of the smooth  functions $\xi_{\mathcal{C}_{\hat{Q}}^{1/2}}^g$, $\xi_{\mathcal{C}_{\hat{P}}^2}^g$ (see Section \ref{subsec-II.2})  with support in  the sets
\begin{equation}
\mathcal{C}_{\hat{Q}}^{1/2}:=\{\hat{k}\,\,:\,\,|\hat{k}\cdot\hat{Q}|\leq \cos(\frac{1}{2}|g|^{\gamma/8})\}\,,
\end{equation}
\begin{equation}
\mathcal{C}_{\hat{P}}^2:=\{\hat{k}\,\,:\,\,|\hat{k}\cdot\hat{P}|\leq\cos(2|g|^{\gamma/8}) \},
\end{equation}
respectively, and from  the constraint $\hat{P}\cdot\hat{Q}\geq \cos(|g|^{\gamma})$.

\noindent
Hence, for  $\vP\in supp f_{\vec{Q}}^g \setminus I^*_{f_{\vec{Q}}^g}$ such that $f_{\vec{Q}}^g(\vP)=1$, we have that
\begin{eqnarray}\label{eq:III.31}
( \Psi_{\vP,E_{\vP}}\,,\,N_{n,\mathcal{C}_{\hat{P}}^2}^b\,\Psi_{\vP,E_{\vP}})&\leq& ( \Psi_{\vP,E_{\vP}}\,,\,N_{n,\mathcal{C}_{\hat{Q}}^{1/2}}^b\,\Psi_{\vP,E_{\vP}})\\
& \leq& \cO(|g|^{(1-2\gamma)}n\| |\vk|^{\beta}\,\one_{(\frac{1}{2(n+1)},\frac{3}{2n})}(\vk)\|_2)\,,
\end{eqnarray}
by definition of $I^*_{f_{\vec{Q}}}$.

\noindent
Because of Eq. (\ref{eq:III.41}), any point $\vP$ belonging to  the set $I^*_{f_{\vec{Q}}^g}$, and such that $f_{\vec{Q}}^g(\vP)=1$,  is at a distance at most 
\begin{eqnarray}
\lefteqn{( \Psi_{f_{\vec{Q}}^g}\,,\,\Psi_{f_{\vec{Q}}^g})^{1/3}\cO(|g|^{\frac{(1-2\gamma)}{3}}n^{\frac{1}{3}}\| |\vk|^{\beta}\,\one_{(\frac{1}{2(n+1)},\frac{3}{2n})}(\vk)\|_2^{\frac{1}{3}})}\\
&\leq&\cO(|g|^{\frac{(1-2\gamma)}{3}}n^{\frac{1}{3}}\| |\vk|^{\beta}\,\one_{(\frac{1}{2(n+1)},\frac{3}{2n})}(\vk)\|_2^{\frac{1}{3}})
\end{eqnarray}
 from an arbitrary  point in $supp f_{\vec{Q}}^g \setminus I^*_{f_{\vec{Q}}^g}$. Thus we consider a slightly modified version of  property (P3) for the operator $N_{n,\mathcal{C}_{\hat{P}}^2}^b$, namely
\begin{equation}
|\vnabla_{\vP}(\Psi_{\vP,E_{\vP}}\,,\,N_{n,\mathcal{C}_{\hat{P}}^2}^b\,\Psi_{\vP,E_{\vP}})|\leq\cO(nC_I[(\sup_{\vP\in I}|E_{\vP}|)+1]|g|^{-\gamma/8})
\end{equation} 
where, following the derivation of property (P3),  the term $|g|^{-\gamma/8}$ comes from the derivative of the smooth function $\xi_{\mathcal{C}_{\hat{P}}^2}^g$. 
Using the fundamental theorem of calculus,  we can finally state that
\begin{equation}\label{eq:IV.81}
( \Psi_{\vP,E_{\vP}}\,,\,N_{n,\mathcal{C}_{\hat{P}}^2}^b\,\Psi_{\vP,E_{\vP}})\leq \cO(|g|^{\frac{(1-2\gamma)}{3}}|g|^{-\gamma/8}n^{\frac{4}{3}}\| |\vk|^{\beta}\,\one_{(\frac{1}{2(n+1)},\frac{3}{2n})}(\vk)\|_2^{\frac{1}{3}}|),\,\, \vP\in I^{*}_{f_{\vec{Q}}^g}.
\end{equation}
We remark that the bounds in Eqs. (\ref{eq:III.31}), (\ref{eq:IV.81}) hold uniformly in $\vec{Q}$,  $|\vec{Q}|\in I'_{g}$.  The bounds in Eqs. (\ref{eq:III.31}), (\ref{eq:IV.81})  hold for $\vP\equiv \vec{Q}$,  because $f_{\vec{Q}}(\vec{Q})=1$ by definition. Thus,
 we  arrive at  the estimate in  Eq. (\ref{eq:III.36}) for any $\vP\in I'_g$.

\noindent
Now, assume that for $(\vP_*, E_{\vP_*})\in I'_g\times \Delta_I$, we have $||\vnabla E_{\vP_*}|-1|>\frac{3}{2}|g|^{\gamma/3}\,$. Then we can consider a wave function $f_{\vec{Q}}^g$ with $\vec{Q}\equiv \vP_*$. Thanks to Property P2,  and for $|g|$ small enough, i.e., less than some value $|\bar{g}|$ uniform in $\vP_*$, $|\vP_*|\in I$, we have that $||\vnabla E_{\vP}|-1|> |g|^{\gamma/3}\,,$ for all $\vP\in f_{\vP_{*}}^g$.  Thus, we can apply Theorem  \ref{theo-III.2}. Finally, 
following the same steps used before, one arrives at the inequality in Eq. (\ref{eq:III.38b}) for $\vP\equiv \vP_*$. Notice that, in this case, since  there is no angular restriction, no term proportional to $|g|^{-\gamma/8}$  appears on  the RHS of Eq. (\ref{eq:III.38b}).
\QED
The bound in Eq. (\ref{eq:III.38b}) trivially implies the corollary below.
\begin{corollary}\label{cor-V.6}
For $\beta>\betacrit$, and for $(\vP,E_{\vP}) \in  I'_g \times \Delta_I $ with $||\vnabla E_\vP|-1|>\frac{3}{2}|g|^{\gamma/3}$, the putative eigenvector  $\Psi_{\vP,E_{\vP}}$ (up to a suitable phase) is asymptotic to the vacuum vector  $\Psi_{\vP}^0$  in $\cH_{\vP}$,  as $g$ tends to $0$. Likewise,  the energy $E_{\vP}$ is asymptotic to $\vP^2/2$. More precisely,
\begin{equation}\label{corollary}
\|\Psi_{\vP}^0-\Psi_{\vP,E_{\vP}}\|\leq \cO(|g|^{(1-2\gamma)/6})
\end{equation}
and
\begin{equation}\label{eq:III.34}
\big|\frac{\vP^2}{2}-E_{\vP}\big|\leq \cO(|g|^{(1-2\gamma)/6})\,.
\end{equation}
\end{corollary}

\noindent
\emph{Proof} 

\noindent
The norm estimate in (\ref{corollary}) follows  from Theorem \ref{theo-III.3}. Without loss of generality, we can start from the identity below,  for some \emph{real} and \emph{positive} coefficient $c(g)$,
\begin{equation}
\Psi_{\vP,E_{\vP}}=c(g)\Psi_{\vP}^0+\Psi_{\vP,E_{\vP}}^{(\geq1)}
\end{equation}
where $\Psi_{\vP,E_{\vP}}$ and $\Psi_{\vP}^0$ are normalized,  and $\Psi_{\vP,E_{\vP}}^{(\geq1)}$ contains at least one boson. Then we can write:
\begin{eqnarray}
\|\Psi_{\vP,E_{\vP}}-\Psi_{\vP}^0\|^2&=&(c(g)-1)^2+\|\Psi_{\vP,E_{\vP}}^{(\geq1)}
\|^2\\
&=&c(g)^2+1-2c(g)+\|\Psi_{\vP,E_{\vP}}^{(\geq1)}
\|^2\,.
\end{eqnarray}
Using the normalization condition,
\begin{equation}
\|\Psi_{\vP,E_{\vP}}\|^2=1=c(g)^2+\|\Psi_{\vP,E_{\vP}}^{(\geq1)}\|^2\,,
\end{equation}
we have
\begin{equation}
c(g)=|1-\|\Psi_{\vP,E_{\vP}}^{(\geq1)}\|^2|^{1/2}
\end{equation}
and
\begin{equation}
\|\Psi_{\vP,E_{\vP}}-\Psi_{\vP}^0\|^2=2-2c(g)\,.
\end{equation}
From Theorem \ref{theo-III.3}, it follows that
\begin{equation}
\|\Psi_{\vP,E_{\vP}}^{(\geq1)}\|^2\leq \|(N^b)^{1/2}\,\Psi_{\vP,E_{\vP}}^{(\geq1)}\|^2\leq \cO(|g|^{(1-2\gamma)/3})\,,
\end{equation}
since the sum over $n$ in \eqref{eq:III.38b} can be estimated as
\begin{equation}
\sum_{n \geq 1} n^{\frac{4}{3}}\||\vk|^{\beta}\,\one_{(\frac{1}{2(n+1)},\frac{3}{2n})}(\vk)\|_2^{\frac{1}{3}}  \leq \sum_{n \geq 1}    n^{\frac{4}{3}}  n^{-\frac{2\beta+3}{6}} \leq const.,  \qquad  \textrm{if}\, \beta > \betacrit
\end{equation}
where we used that $\| |\vk|^{\beta}\,\one_{(\frac{1}{2(n+1)},\frac{3}{2n})}(\vk)\|_2 = \cO(n^{-\frac{2\beta+3}{2}} ) $, as follows by  the size $\cO(1/n)$ of the support of the function $\chi_n$, and the spatial dimension $d=3$. 
We remind the reader that the expectation value in $\Psi_{\vP,E_{\vP}}$ of the number operator associated with boson momenta above $|\vk|=1$ can be bounded above by using the form inequality $H^f<aH_{\vP}+b$, for some $a,b>0$.

Consequently,  the estimate in Eq. (\ref{corollary}) is easily obtained.

\noindent
For the inequality in Eq. 
(\ref{eq:III.34}) consider
\begin{eqnarray}
E_{\vP}-\frac{\vP^2}{2}&= &(\Psi_{\vP,E_{\vP}}, H_{\vP}(\Psi_{\vP,E_{\vP}}-\Psi_{\vP}^0))\\
& &+(\Psi_{\vP,E_{\vP}}, (H_{\vP}-H_{\vP}^0)\Psi_{\vP}^0)\\
& &+(\Psi_{\vP,E_{\vP}}-\Psi_{\vP}^0, H_{\vP}^0\Psi_{\vP}^0)\,.
\end{eqnarray}
Then use (\ref{eq:III.34}) and the fact that $H_{\vP}-H_{\vP}^0= g \phi^b(\rho)$ is $H_{\vP}^{0}$-bounded.
\QED
\\

\section{Absence of regular mass shells}\label{Section-VI}
\resetequ

In this section, we first make use of the results obtained in Section \ref{Section-V} to arrive at an argument  that shows a contradiction to the existence of a mass shell $(\vP,E_{\vP})\in I_g\times\Delta_I $ assuming that $||\vnabla E_{\vP}|-1|>\frac{3}{2}|g|^{\gamma/3}$ for $\vP\in I'_g$. In implementing the argument we  employ suitable trial states; see  Theorem \ref{theo-VI.1}. Then we proceed to showing that, if we remove the assumption $||\vnabla E_{\vP}|-1|>\frac{3}{2}|g|^{\gamma/3}$,  a mass shell might exist for $(\vP,E_{\vP})\in I_g\times\Delta_I$ such that 
\begin{equation}
 E_{\vP}=|\vP|-\frac{1}{2}+\cO(|g|^{\gamma/4})\,.
\end{equation}
This result is completed in Theorem \ref{theo-VI.4}.

\noindent
We recall that so far we have assumed the existence of a mass shell for $\vP$ in the open interval $I_g$,  and we have defined another open interval $I'_g \subset I_g$ with the properties specified in Section \ref{subsec-II.2}. The results of Corollary \ref{cor-V.6}, which will be used in the following theorem, hold for $\vP\in I'_g$.
\begin{theorem} \label{theo-VI.1}
For $\beta>\betacrit$, and for $|g|$ small enough,  
no regular (i.e., fulfilling  the \emph{Main Hypothesis} in Section \ref{sec: main hypothesis})  mass shell $(\vP,E_{\vP})$ can exist with the properties:
\begin{itemize}
\item [i)] $|\vP|\in I_g$, $|I_g|>|g|^{\gamma/2}$;
\item [ii)] $|E_{\vP}| \in \Delta_I$;
\item [iii)] $||\vnabla E_{\vP}|-1|>\frac{3}{2}|g|^{\gamma/3}$ for $\vP\in I'_g$.
\end{itemize}
\end{theorem}
\emph{Proof}
The proof is by contradiction. For $|g|$ sufficiently small (depending on the exponent $\gamma$), we pick an open interval $I''_g \subset I'_g$ fulfilling the following properties: 
\begin{itemize}
\item[(a)]  $|I''_g|>|g|^{\gamma}$; 
\item[(b)] If $|\vec{Q}|\in I''_g$  then  $|\vP|\in I'_g$ for any $\vP \in supp f_{\vec{Q}}^g$.  
\end{itemize}
Notice that the definition of $I''_g$ is meaningful for $|g|$ small enough. For $|\vec{Q}| \in I''_g$,  we introduce the trial vector
\begin{equation}
\trialvector\,:=\,\int d^3P\int_{}\,d^3k\, f^g_{\vec{Q}}(\vP)\,\frac{1}{\epsilon^{\frac{1}{2}}}\,h\big(\frac{(\vP-\vk)^2/2+|\vk|-E_{\vP}}{\epsilon}\big)b^{*}_{\vk}\,\Psi_{\vP}^{0}\,,\label{eq: trial vector}
\end{equation}
where:
\begin{itemize}
 \item
 $\epsilon>0$; 
  \item
 $h(z)\in C_0^{\infty}(\RR)$, $h(z)\geq0$. 
 \end{itemize}
\noindent 

\noindent
Since  $\Psi_{f^g_{\vec{Q}}}$ is a single-particle state, we have that
\begin{equation}
(\trialvector\,,\,(H^{0}-E_{\vP})\Psi_{f^g_{\vec{Q}}})=-(\trialvector\,,\,g\phi (\rho_{\vx})\,\Psi_{f^g_{\vec{Q}}})\,,
\end{equation}
where $H^{0}:=\frac{\vp^2}{2}+H^f$ and $E_{\vP}$ is a (operator-valued) function of the total momentum operator $\vP$.
This equation implies that
\begin{eqnarray}
\lefteqn{g\,(\trialvector,\,\phi(\rho_{\vx})\,P_{\Omega}\Psi_{f^g_{\vec{Q}}})}\\
 & =&-(\trialvector,\,(H^{0}-E_{\vP})\, P^{\perp}_{\Omega}\,\Psi_{f^g_{\vec{Q}}})\\
 & &-g\,(\trialvector,\,\phi(\rho_{\vx})\, P_{\Omega}^{\perp}\Psi_{f^g_{\vec{Q}}})\,,
 \end{eqnarray}
 where, as usual, the expressions $P_{\Omega}$, $P_{\Omega}^{\perp}$ acting on $\cH$ stand for $\one_{\cH_{el}}\otimes P_{\Omega}$, $\one_{\cH_{el}}\otimes P_{\Omega}^{\perp}$, respectively.
We observe that
\begin{eqnarray}
\lefteqn{(\trialvector,\,\phi(\rho_{\vx})\, P_{\Omega}\Psi_{f^g_{\vec{Q}}})}\\
&= &c(g)\,\int\,d^3P \int d^3k |f^g_{\vec{Q}}(\vP)|^{2}\,\frac{1}{\epsilon^{\frac{1}{2}}}\,h\big(\frac{(\vP-\vk)^2/2+|\vk|-E_{\vP}}{\epsilon}\big)\, \rho(|\vk|)\,,\quad\quad\quad\nonumber
\end{eqnarray}
where $c(g)\to 1$,  as $g\to 0$, because of Corollary \ref{cor-V.6}. Notice that, for $|\vP|>1+\delta$, where $\delta>0$ is $g$-independent,  the equation
 \begin{equation}\label{eq:III.81}
(\vP-\vk)^2/2+|\vk|-\vP^2/2=0,\,  \qquad  \str \vk \str >0
\end{equation}
has the one-parameter family of solutions $$|\vk|=2(|\vP|\cos\theta-1)> 0 $$ for $\cos(\theta)-\frac{1}{|\vP|} >0$,  where  $\cos\theta=\frac{\vP\cdot\vk}{|\vP||\vk|}$. 

\noindent
Notice that,  for $\vP\in I$, $\rho(2(|\vP|\cos\theta-1))\neq 0$ for some $0<\theta< \pi$, see the conditions on $\rho$ in Section \ref{sec: hamiltonians}.
\noindent
Hence, using (\ref{eq:III.34}), for $\epsilon$ and $|g|$ small enough, we arrive at the following bound
\begin{equation}
|(\trialvector,\,\phi(\rho_{\vx})\, P_{\Omega}\Psi_{f^g_{\vec{Q}}})|>D_1\,\epsilon^{\frac{1}{2}}\|f^g_{\vec{Q}}\|_2^2\,,\label{eq:III.47}
\end{equation}
where $D_1$ is an $\epsilon$- and $g$- independent (positive) constant; (hint: for each $\theta$ in Eq. (\ref{eq:III.81}), implement the change of variable $|\vk|\to z_{\theta}$ with  $z_{\theta}:=[(\vP-\vk)^2/2+|\vk|-E_{\vP}]/\epsilon$).

 \noindent
Using the Schwarz inequality, we find that
\begin{equation}
|(\trialvector,N^{\frac{1}{2}}\,(H^{0}-E_{\vP})\, P^{\perp}_{\Omega}\, \Psi_{f^g_{\vec{Q}}})| 
\end{equation}
\begin{equation}
  \leq  \|(H^{0}-E_{\vP})\trialvector\| \| N^{\frac{1}{2}}\,P^{\perp}_{\Omega}\,\Psi_{f^g_{\vec{Q}}}\| \,.
\end{equation}
We then observe that
\begin{equation}
\|(H^{0}-E_{\vP})\trialvector\|\leq\cO(\|f^g_{\vec{Q}}\|_2\,\epsilon)\,.
\end{equation}
Using Eq. \eqref{eq:III.14a},  one may easily derive  the inequalities
\begin{equation}
\|N^{\frac{1}{2}} P^{\perp}_{\Omega}\, \Psi_{f^g_{\vec{Q}}}\|\,\leq \cO(|g|^{\frac{2(1-2\gamma)}{2}}\|f^g_{\vec{Q}}\|_2)\,
\end{equation}
and
\begin{eqnarray}
\lefteqn{|(\trialvector,\,\phi(\rho_{\vx})\,P_{\Omega}^{\perp} \Psi_{f^g_{\vec{Q}}})|  } \\
& =&|(\trialvector,\,N\,\phi(\rho_{\vx})\,  P_{\Omega}^{\perp} \Psi_{f^g_{\vec{Q}}})|\label{eq:trial1}\\
& \leq &  \cO (|g|^{\frac{2(1-2\gamma)}{2}} \|f^{g}_{\vec{Q}}\|_2)\label{eq:trial2}\,. 
\end{eqnarray}
For the step from (\ref{eq:trial1})  to (\ref{eq:trial2}), one may use that 
\begin{equation}
(\trialvector,\,N \,\phi(\rho_{\vx})\,P_{\Omega}^{\perp} \Psi_{f^g_{\vec{Q}}})=(\trialvector,\,N\, \phi^{(-)}(\rho_{\vx})\,P_{\Omega}^{\perp} \Psi_{f^g_{\vec{Q}}})
\end{equation}
where $\phi^{(-)}(\rho_{\vx})$,  $\phi^{(+)}(\rho_{\vx})$ stand for the part proportional to the annihilation- and to the creation operator,  respectively; i.e., $\phi(\rho_{\vx})=\phi^{(-)}(\rho_{\vx})+\phi^{(+)}(\rho_{\vx})$. Then,  we observe that
\begin{eqnarray}
\lefteqn{(\trialvector,\,N  \,\phi^{(-)}(\rho_{\vx})\,P_{\Omega}^{\perp} \Psi_{f^g_{\vec{Q}}})=}\\
& =&(\trialvector,\,\phi^{(-)}(\rho_{\vx})\, N\,P_{\Omega}^{\perp} \Psi_{f^g_{\vec{Q}}})\\
& &-(\trialvector,\,[\phi^{(-)}(\rho_{\vx})\,,\, N]\, P_{\Omega}^{\perp} \Psi_{f^g_{\vec{Q}}}),
\end{eqnarray}
and we finally use Theorem  \ref{theo-III.2} together with the estimates
\begin{equation}
\|N^{\frac{1}{2}}\phi^{(+)}(\rho_{\vx})\trialvector\|\leq \cO(\|f^g_{\vec{Q}}\|_2)\,,
\end{equation}
\begin{equation}
\|[\phi^{(-)}(\rho_{\vx})\,,\, N ]\,P_{\Omega}^{\perp} \Psi_{f^g_{\vec{Q}}}\|\leq \cO(\|N^{\frac{1}{2}}\,P_{\Omega}^{\perp} \Psi_{f^g_{\vec{Q}}}\|)\,.
\end{equation}

\noindent
Finally, we arrive at
\begin{equation}
D_1\,|g|\,\epsilon^{\frac{1}{2}}\|f^g_{\vec{Q}}\|_2^2\,\leq\,\cO(\epsilon\,|g|^{(1-2\gamma)}\|f^g_{\vec{Q}}\|^2_2)+\cO(|g|^{2 - 2\gamma}\|f^g_{\vec{Q}}\|^2_2)\,.
\end{equation}
This inequality  is violated whenever 
\begin{equation}\label{gamma-inequality}
c_1|g|^{1-2\gamma}<\epsilon^{\frac{1}{2}}<c_2|g|^{2\gamma}
\end{equation}
for some $c_1, c_2>0$.  We note that the inequality in Eq. (\ref{gamma-inequality}) can be fulfilled  if $0<\gamma<1/4$ and $|g|$ is small enough.

\noindent
From the argument above, we conclude that, for $|g|$ small enough,  a mass shell cannot exist in $I_g\times \Delta_I$ with the assumed  regularity properties,  because  $I''_g\subset I'_g \subset I_g$.
\QED

We need two preparatory lemmas to state our final result, Theorem \ref{theo-VI.4},  concerning the absence of a mass shell anywhere but near the boundary of the energy-momentum  spectrum. 

\noindent
From property (P1), we know that the vector $\vnabla E_{\vP}$ is collinear to $\vP$. In the first of the two lemmas below, Lemma \ref{cor-III.6},  assuming that $||\vnabla E_{\vP}|-1|\leq \frac{3}{2}|g|^{\gamma/3}$ and $\beta>\betacrit$, we show that $\vnabla E_{\vP}$ and $\vP$ are in fact parallel.

\noindent
The second lemma, Lemma  \ref{cor-III.7}, states that  the boson number operator, restricted to the cone $\{\hat{k}\,:\, -\hat{k}\cdot\hat{P}<\cos(2|g|^{\gamma/8})\}$ and evaluated on the putative fiber eigenvector $\Psi_{\vP,E_{\vP}}$,  $|\vP|\in I'_g$, is also bounded above by $\cO(|g|^{\frac{(1-2\gamma)}{3}}|g|^{-1/8})$,  for $\beta>\betacrit$. 
\begin{lemma}\label{cor-III.6}
For $\beta>\betacrit$, and for $g$ in an  interval $\{g\,:\,0<|g|\leq g_*\}$ with $g_*>0$ small enough, if $(\vP,E_{\vP}) \in  I_g\times \Delta_I$ fulfills the constraint  
\begin{equation}\label{eq:constraint2}
||\vnabla E_{\vP}|-1|\leq \frac{3}{2}|g|^{\gamma/3}\,
\end{equation}
then  the bound $\frac{\partial E_{\vP}}{\partial |\vP|}\geq1 -\frac{3}{2}|g|^{\gamma/3}$ holds true.
\end{lemma}

\noindent
\emph{Proof}

\noindent
The proof is indirect. We assume that there exists $g_*>0$ such that,  for some $|g|<g_*$ and for some $\vP_*\in I_g$,
\begin{equation}\label{assumption}
\frac{\partial E_{\vP}}{\partial |\vP|}|_{\vP=\vP_*}<-1 +\frac{3}{2}|g|^{\gamma/3}<0\,.
\end{equation}
We also assume that $g_*$ is small enough to apply Lemma  \ref{lem-V.3} and Theorem \ref{theo-III.3} later on.
We shall show that the assumption in Eq.  (\ref{assumption})  yields a contradiction.  Consider the function $f_{\vec{Q}\equiv \vP_*}^g$. By Property P2 
\begin{equation}
\frac{\partial E_{\vP}}{\partial |\vP|}<c|g|^{\gamma}, \qquad \textrm{with}\,  c>0,
\end{equation}
for all $\vP\in supp f_{\vec{Q}\equiv \vP_*}^g$.  Now, for all $\hat{u}$-dependent sectors such that $$\hat{u}\cdot \hat{P_*}>0\,,$$ we consider  the first virial identity of Section \ref{Section-V.I.1} (see Eqs. (\ref{eq:III.12})-(\ref{eq:III.20c})) and observe that
\begin{eqnarray}
\lefteqn{-(\Psi_{f_{\vec{Q}}^g}\,,\,\vnabla E_{\vP}\cdot d\Gamma(\chi_{n}^2(|\vk|)\xi^{g\,2}_{\hat{u}}(\hat{k})\,\vk)\Psi_{f_{\vec{Q}}^g})}\\
&\geq& -c|g|^{\gamma}\,(\Psi_{f_{\vec{Q}}^g}\,,\,d\Gamma(\chi_{n}^2(|\vk|)\xi^{g\,2}_{\hat{u}}(\hat{k})\,|\vk|)\Psi_{f_{\vec{Q}}^g})\,,
\end{eqnarray}
for all $\vP\in supp f_{\vec{Q}\equiv \vP_*}^g$. 
Then, for $|g|<g_*$ and $g_*$ small enough,  one can  proceed as in Lemma  \ref{lem-V.3}, and finally apply the argument used in Theorem \ref{theo-III.3} to obtain that 
\begin{equation}
( \Psi_{\vP_*,E_{\vP_*}}\,,\,N_{n, \hat{u}}^b\,\Psi_{\vP_*,E_{\vP_*}})\, 
\end{equation}
 can be summed  over $n$, yielding a quantity bounded by $\cO(|g|^{\gamma/2})$.  This result  readily implies that
\begin{equation}
( \Psi_{\vP_*,E_{\vP_*}}\,,\,\vP^{f}\,\Psi_{\vP_*,E_{\vP_*}})\cdot \hat{P_*}\leq C |g|^{\gamma/2}\,
\end{equation}
for some positive constant $C$, hence
\begin{equation}
-( \Psi_{\vP_*,E_{\vP_*}}\,,\,\vP^{f}\,\Psi_{\vP_*,E_{\vP_*}})\cdot \hat{P_*}\geq- C |g|^{\gamma/2}\,
\end{equation}

\noindent
Using the Feynman-Hellman formula 
\begin{equation}
\vnabla E_{\vP}=\frac{\partial E_{\vP}}{\partial |\vP|} \hat{P}=\vP-( \Psi_{\vP,E_{\vP}}\,,\,\vP^{f}\,\Psi_{\vP,E_{\vP}})\,,
\end{equation}
we deduce that 
\begin{equation}
\frac{\partial E_{\vP}}{\partial |\vP|}|_{\vP=\vP_*}\geq |\vP_*|-C |g|^{\gamma/2}>1-C |g|^{\gamma/2}\,.
\end{equation}
This yields a a contradiction for $g_*$ small enough, therefore we conclude that the bound 
\begin{equation}
\frac{\partial E_{\vP}}{\partial |\vP|}|_{\vP=\vP_*}\geq1 -\frac{3}{2}|g|^{\gamma/3}
\end{equation}
holds for $\{g\,|\,0<|g| \leq g^*\}$,  for some $g^*>0$, because of (\ref{eq:constraint2}) .
 \QED

We are now in a position to extend the result in Eq. (\ref{eq:III.36}).
\begin{lemma}\label{cor-III.7}
For  $(\vP,E_{\vP}) \in  I'_g\times \Delta_I$, with $||\vnabla E_{\vP}|-1|\leq \frac{3}{2} |g|^{\gamma/3}$, and for $\beta>\betacrit$ and $|g|$ small enough,
\begin{equation}\label{eq:con6}
( \Psi_{\vP,E_{\vP}}\,,\,N_{n,\mathcal{C}^2_{\hat{P}}\,\cup \,\mathcal{C}_{\hat{P}}^{2,-}}^b\,\Psi_{\vP,E_{\vP}})\leq \cO(|g|^{\frac{(1-2\gamma)}{3}}|g|^{-1/8}n^{\frac{4}{3}}\| |\vk|^{\beta}\,\one_{(\frac{1}{2(n+1)},\frac{3}{2n})}(\vk)\|_2^{\frac{1}{3}})\,,
\end{equation}
where
\begin{equation}
N_{n,\mathcal{C}^2_{\hat{P}}\,\cup \,\mathcal{C}_{\hat{P}}^{2,-}}^b:=d\Gamma^b(\chi_{n}^2(|\vk|)\,\xi_{\mathcal{C}^{g\,2}_{\hat{P}}\,\cup \,\mathcal{C}_{\hat{P}}^{2,-}}^{g\,2}(\hat{k}))\,
\end{equation}
and  $\xi^{g}_{\mathcal{C}_{\hat{P}}^2\,\cup \,\mathcal{C}_{\hat{P}}^{2,-}}(\hat{k})$, $0\leq \xi^{g}_{\mathcal{C}_{\hat{P}}^2\,\cup \,\mathcal{C}_{\hat{P}}^{2,-}}(\hat{k})\leq 1$,  is a smooth function with support in
\begin{equation}
\mathcal{C}^2_{\hat{P}}\,\cup \,\mathcal{C}_{\hat{P}}^{2,-}\,,
\end{equation}
where $\mathcal{C}_{\hat{P}}^{2,-}:=\{\hat{k}\,:\,-\hat{k}\cdot\hat{P}\geq \cos(2|g|^{\gamma/8})\}$. $\xi^{g}_{\mathcal{C}_{\hat{P}}^2\,\cup \,\mathcal{C}_{\hat{P}}^{2,-}}(\hat{k})$ is  defined as follows
 \begin{itemize}
\item[i)]
 \begin{equation}
 \xi_{C^2_{\hat{P}}\cup C^{2,-}_{\hat{P}}}^g(\hat{k})=1\quad\text{for}\quad\{\hat{k}\,\,:\,\,\hat{k}\cdot\hat{P}\leq \cos(4|g|^{\gamma/8})\}\,;
 \end{equation}
\item[ii)] 
 \begin{equation}\xi_{C^2_{\hat{P}}\cup C^{2,-}_{\hat{P}}}^g(\hat{k})=0\quad\text{for}\quad\{\hat{k}\,\,:\,\,\hat{k}\cdot\hat{P}> \cos(2|g|^{\gamma/8})\}\,;
 \end{equation}
\item[iii)] 
\begin{equation}
|\partial_{\theta_{\hat{kP}}} \xi_{C^2_{\hat{P}}\cup C^{2,-}_{\hat{P}}}^g(\hat{k})|\leq C_{\xi}\,|g|^{-\gamma/8}\,,\label{eq:function-cone in thm}
\end{equation}
where $\theta_{\hat{kP}}$ is the angle between $\hat{k}$ and $\hat{P}$, and the constant $C_{\xi}$ is independent of $g$.
\end{itemize}

\end{lemma}

\noindent
\emph{Proof}

\noindent
Because of Lemma \ref{cor-III.6}, for $\hat{k}$ in the sector $\mathcal{C}_{\hat{Q}\equiv \vP}^{2,-}$ and $\vP'\in supp f_{\vec{Q}\equiv \vP}^g$, with $\vec{Q}\equiv \vP\in I'_g$, the condition in (\ref{eq:III.26-bis}) of Lemma   \ref{lem-V.3} is fulfilled. Then one can repeat the arguments of Theorem \ref{theo-III.3} for the number operator restricted to the sector  $\mathcal{C}_{\hat{Q}\equiv \hat{P}}^{2,-}$, and derive the inequality in Eq. (\ref{eq:con6}) for all $\vec{Q}\equiv\vP\in I'_g$.  \QED
\\

\begin{theorem}\label{theo-VI.4}
For $\beta>\betacrit$  and $|g|$ small enough, if a regular mass shell (i.e.\ fulfilling the \emph{Main Hypothesis}  in Section \ref{sec: main hypothesis}) exists in an interval $I_g$,  and if for some  $(\vP,E_{\vP})\in  I'_g\times \Delta_I $
\begin{equation}\label{constraint}
||\vnabla E_\vP|-1| \leq \frac{3}{2}|g|^{\gamma/3}\,,
\end{equation} 
 then, for all $\vP \in I_g$,
\begin{equation}
 E_{\vP}=|\vP|-\frac{1}{2}+\cO(|g|^{\gamma/4})\,.
\end{equation}
\end{theorem}

\noindent
\emph{Proof}

\noindent
We consider $(\vP, E_{\vP})\in I'_g\times \Delta_I$ such that
\begin{equation}\label{eq:in}
||\vnabla E_\vP|-1| \leq \frac{3}{2}|g|^{\gamma/3}\,.
\end{equation}
From the Feynman-Hellman formula (see Eq.(\ref{eq:II.5.1})) 
\begin{equation}\label{eq:F-H1}
\vP\cdot\vnabla E_{\vP}\,=\,|\vP|^2-\vP\cdot(\Psi_{\vP,E_{\vP}}\,,\,\vP^f\,\Psi_{\vP,E_{\vP}})\,.
\end{equation}
From the result  in Lemma \ref{cor-III.6}, 
we can derive the following identity 
\begin{equation}\label{eq:F-H2}
\vP\cdot\vnabla E_{\vP}\,=\,|\vP|(1+\,\cO(|g|^{\gamma/3}))\,.
\end{equation}

\noindent
From Lemma \ref{cor-III.7},  for  the expectation values in the equation below,  we can restrict $\vP^f$ and $H^f$ to the sector $\mathcal{C}_{\hat{P}}^{2,+}:=\{\hat{k}\,:\,\hat{k}\cdot\hat{P}\geq \cos(2|g|^{\gamma/8})\}$ up to an $o((|g|^{\gamma/4})$ remainder, and we deduce that
\begin{equation}\label{eq:F-H3}
\hat{P}\cdot(\Psi_{\vP,E_{\vP}}\,,\,\vP^f\,\Psi_{\vP,E_{\vP}})=(\Psi_{\vP,E_{\vP}}\,,\,H^f\,\Psi_{\vP,E_{\vP}})+\cO(|g|^{\gamma/4})\,.
\end{equation}
Hence, by combining  (\ref{eq:F-H1})-(\ref{eq:F-H3}), one arrives at
\begin{eqnarray}
\lefteqn{(\Psi_{\vP,E_{\vP}}\,,\,H^f\,\Psi_{\vP,E_{\vP}})-\vP\cdot(\Psi_{\vP,E_{\vP}}\,,\,\vP^f\,\Psi_{\vP,E_{\vP}}) \nonumber}\\
&= &|\vP|-1+|\vP|-|\vP|^2+\cO(|g|^{\gamma/4})\,.
\end{eqnarray}
Next, starting from the \emph{formal} virial identity
\begin{equation}
(\Psi_{\vP,E_{\vP}}\,,\,i[H_{\vP}\,,\,D^b_{\frac{1}{\kappa}, \kappa}]\,\Psi_{\vP,E_{\vP}})=0\,,
\end{equation}
where $ D_{\frac{1}{\kappa}, \kappa}^{b}=d\Gamma^b(d_{\frac{1}{\kappa}, \kappa})$ is defined in Section \ref{subsec-II.2} ($\cI\cI\cI$), we derive
\begin{eqnarray}
0&= &(\Psi_{\vP,E_{\vP}}\,,\,d\Gamma^b(i[|\vk|,\,d_{\frac{1}{\kappa}, \kappa}])\Psi_{\vP,E_{\vP}})\label{final-virial}\\
& &+( \Psi_{\vP,E_{\vP}}\,,\,d\Gamma^b(i[\vk,\,d_{\frac{1}{\kappa}, \kappa}])\cdot d\Gamma^b(\vk)\Psi_{\vP,E_{\vP}})\nonumber\\
& &-\vP\cdot (  \Psi_{\vP,E_{\vP}}\,,\,d\Gamma^b(i[\vk,\,d_{\frac{1}{\kappa}, \kappa}])\Psi_{\vP,E_{\vP}}) \nonumber \\
& &-g(\Psi_{\vP,E_{\vP}}\,,\,[b^*(id_{\frac{1}{\kappa}, \kappa} \,\rho)+b(id_{\frac{1}{\kappa}, \kappa} \,\rho)] \Psi_{\vP,E_{\vP}})\,.\nonumber
\end{eqnarray}
The virial identity in Eq. \eqref{final-virial} needs to be justified and this is done in Section \ref{sec: proof of virial fiber} in the Appendix.

By taking the limit $\kappa \uparrow +\infty$ on the RHS of  \eqref{final-virial}, it follows that
\begin{eqnarray}
0&= &(\Psi_{\vP,E_{\vP}}\,,\,H^f\Psi_{\vP,E_{\vP}})\label{final-virial infinity}\\
& &+( \Psi_{\vP,E_{\vP}}\,,\,\vec{P}^f\cdot \vec{P}^f\Psi_{\vP,E_{\vP}})\nonumber\\
& &-\vP\cdot (  \Psi_{\vP,E_{\vP}}\,,\,\vec{P}^f\Psi_{\vP,E_{\vP}}) \nonumber \\
& &-g(\Psi_{\vP,E_{\vP}}\,,\,[b^*(id_{\infty} \,\rho)+b(id_{\infty} \,\rho)] \Psi_{\vP,E_{\vP}})\,,\nonumber
\end{eqnarray}
where
\begin{equation}\label{eq:III.7 infty}
d_{\infty}\,:=\,\frac{1}{2}\,(\vk\cdot\,i\vnabla_{\vk}+i\vnabla_{\vk}\cdot\,\vk)\,.
\end{equation}
Eq. \eqref{final-virial infinity} follows from  \eqref{final-virial} thanks to
\begin{enumerate}
\item the infrared behavior of the form factor $\rho(\vk)$, namely for any $\beta>-1$.
\item  the ultraviolet cut-off $\Lambda$; see Eq. (\ref{assumption on form factor}).
\item  the fact that $d\Gamma^b(i[|\vk|,\,d_{\frac{1}{\kappa}, \kappa}])$ and $d\Gamma^b(i[\vk,\,d_{\frac{1}{\kappa}, \kappa}])$ are bounded by $H^f$ and $\Psi_{\vP,E_{\vP}}$ belongs to the domain of $H^f$.
\end{enumerate}
Therefore,  we can express the expectation value of $(\vP^{f})^2$ in the state $\Psi_{\vP,E_{\vP}}$ as a function of $|\vP|$ up to $g$-dependent corrections
\begin{equation}
( \Psi_{\vP,E_{\vP}}\,,(\vP^{f})^2\,\Psi_{\vP,E_{\vP}})=(|\vP|-1)^2+
\cO(|g|^{\gamma/4})\,.
\end{equation}
Using the eigenvalue equation (\ref{eq:III.20bis}), we obtain
\begin{eqnarray}\label{eq:fin}
E_{\vP}
&=&\frac{1}{2}\big[(|\vP|-1)^2+2|\vP|-|\vP|^2+
\cO(|g|^{\gamma})\big]+|\vP|-1+\cO(|g|^{\gamma/4})\nonumber\\
&= &|\vP|-\frac{1}{2}+\cO(|g|^{\gamma/4})\,.\label{eq:final}
\end{eqnarray}
Finally, because of the constraint on $\vnabla E_{\vP}$  (see Property P1,  Section \ref{subsec-II.1}),  if Eq. (\ref{eq:final}) holds for $|\vP|\in I'_g$, either it is also true for $|\vP|\in I_g$ or the mass shell cannot be defined on $I_g$ with the assumed regularity properties. This can be explained considering the following two cases:
\begin{itemize}
\item [a)]
if $|I_g|<2|g|^{\gamma/4}$, use that $|\vnabla E_{\vP}|<C'_I$ and conclude that Eq. (\ref{eq:fin}) holds on $I_g$;
\item[b)]
if $|I_g| \geq 2|g|^{\gamma/4}$,  write $I_g$ as $I_g=\cup_jI^j_g$, with $\{I^j_g\}$ disjoints,  and $2|g|^{\gamma/4}>|I^{j}_g|>|g|^{\gamma/2}$. For each $I^j_g$,  either one can repeat the argument developed  in Eqs. (\ref{eq:in})-(\ref{eq:fin}), and proceed as in a),  or one concludes that the mass shell does not exists for $\vP\in I^j_g$. In the latter case, since $I^j_g\subset I_g$, the mass shell does not exist in $I_g$ with the assumed regularity properties.
\end{itemize}

\noindent
{\bf{Remark}}

\noindent
It is easy to see that 
\begin{equation}\label{eq:energy-constraint}
E_{\vP}\leq \frac{\vP^2}{2}+\cO(|g|)\,.
\end{equation}
The proof follows from Eq. (\ref{final-virial infinity}) by adding and subtracting $\vP^2$ on the right-hand side. In fact, one gets
\begin{eqnarray}\label{eq:virial2}
0&= &(\Psi_{\vP,E_{\vP}}\,,\,H_{\vP} \Psi_{\vP,E_{\vP}})-\frac{\vP^2}{2}\\
& &+\frac{1}{2}( \Psi_{\vP,E_{\vP}}\,,\,\vec{P}^f\cdot \vec{P}^f\Psi_{\vP,E_{\vP}})\nonumber\\
& &-g(\Psi_{\vP,E_{\vP}}\,,\,[b^*(id_{\infty} \,\rho)+b(id_{\infty} \,\rho)] \Psi_{\vP,E_{\vP}})\,.\nonumber
\end{eqnarray}

Furthermore, assuming the validity of the Feynman-Helman formula, we see that
\begin{eqnarray}
\vP \cdot \vnabla E_{\vP} &=&  \vP \cdot (\Psi_{\vP, E_{\vP}},(\vP- \vP^f) \Psi_{\vP, E_{\vP}})\\
&\geq& \vP^2-|\vP|\|\vP^f\Psi_{\vP, E_{\vP}}\|
\end{eqnarray}
From Eq. (\ref{eq:virial2})
\begin{equation}
\|\vP^f\Psi_{\vP, E_{\vP}}\|^2\leq \vP^2-2E_{\vP}+C|g|\quad,\quad C>0,
\end{equation}
and then
\begin{equation}
\vP\cdot \vnabla E_{\vP}\geq \vP^2-|\vP|\sqrt{\vP^2-2E_{\vP}+C|g|}
\end{equation}
For $|\vP|\geq 1+\delta$, because of the constraint $E_{\vP}\geq |\vP|-\frac{1}{2}+\cO(|g|)$,  we can conclude that
\begin{equation}
\hat{P}\cdot \vnabla E_{\vP}\geq 1-C'|g|
\end{equation}
for some positive constant $C'$. This yields an alternative proof of Lemma \ref{cor-III.6}.

\section{Appendix }\label{Appendix}
\resetequ

In Sections \ref{appsec: lemma x domain} and \ref{appsec: lemma virial}, we provide the proofs of Lemmas  \ref{lem-III.1} and \ref{lem-III.2} in Section \ref{Section-V}. For the convenience of the reader, these lemmas are repeated below.  In Section \ref{sec: proof of virial fiber}, we prove the equality \eqref{final-virial} in Section \ref{Section-VI}. 

Lemma \ref{lem-III.2} and the equality \eqref{final-virial} are virial identities whose justification is, in general, a hard task. We refer the reader to \cite{cyconfroesekirschsimon, froehlichgriesemersigal} and \cite{georgescugerard} for more background.

\subsection{Proof of Lemma \ref{lem-III.1}} \label{appsec: lemma x domain}

\begin{blanklemma}[\textrm{\ref{lem-III.1}}]
The vector $ \Psi_{f_{\vec{Q}}^g}$ belongs to the domain of the position operator $\vx$ and 
\begin{equation}\label{eq:V.2 app}
 \norm x_i \Psi_{f_{\vec{Q}}^g}\| \leq \cO(|g|^{-\gamma} \|\Psi_{f_{\vec{Q}}^g}\|),\, \qquad  i=1,2,3
\end{equation}
\end{blanklemma}
\begin{proof}
It suffices to estimate, in the limit $\Delta_i\to0$,
\begin{eqnarray}
\lefteqn{\frac{e^{-i  \Delta_i x_i}\Psi_{f_{\vec{Q}}^g} - \Psi_{f_{\vec{Q}}^g}}{\Delta_i}}\\
& =&\frac{1}{\Delta_i}\big[e^{-i  \Delta_i x_i}\int f_{\vec{Q}}^g(\vP)\,\Psi_{\vP,E_{\vP}}d^3P-\int f_{\vec{Q}}^g(\vP)\,\Psi_{\vP,E_{\vP}}d^3P\big] \nonumber\\
&= &\frac{1}{\Delta_i}\big[\int f_{\vec{Q}}^g(\vP)\,      e^{-i  \Delta_i x_i}\Psi_{\vP,E_{\vP}} d^3P-\int f_{\vec{Q}}^g(\vP)\,\Psi_{\vP-\Delta_i \hat{i},E_{\vP-\Delta_i \hat{i}}}d^3P\big]  \label{eq:III.22}\\
& &+ \frac{1}{\Delta_i}\big[\int \,(f_{\vec{Q}}^g(\vP)- f_{\vec{Q}}^g(\vP-\Delta_i \hat{i}))\,\Psi_{\vP-\Delta_i \hat{i},E_{\vP-\Delta_i \hat{i}}}d^3P\big] \label{eq:III.23}\\
& &+\frac{1}{\Delta_i}\big[\int f_{\vec{Q}}^g(\vP-\Delta_i \hat{i})\,\Psi_{\vP-\Delta_i \hat{i},E_{\vP-\Delta_i \hat{i}}}d^3P-\int f_{\vec{Q}}^g(\vP)\,\Psi_{\vP,E_{\vP}}d^3P\big] \,\quad\quad \quad\label{eq:III.24}
\end{eqnarray}
We notice that $ e^{-i  \Delta_i x_i}\Psi_{\vP,E_{\vP}} \in \cH_{\vP-\Delta_i \hat{i}}$ (in \eqref{eq:III.22}), and 
\begin{equation}
I_{\vP-\Delta_i \hat{i}}(     e^{-i  \Delta_i x_i}\Psi_{\vP,E_{\vP}}   )=I_{\vP}(\Psi_{\vP,E_{\vP}})\,
\end{equation}
as vectors in $\cF^b$. The term in (\ref{eq:III.24}) is identically zero, by a change of variables.
We now derive bounds for (\ref{eq:III.22}),  (\ref{eq:III.23}), as $\Delta_{i}\to0$. \\

By  item $(iii)$ in the \emph{Main Hypothesis}  (which, strictly speaking, means that $\| \vnabla_{\vP} I_{\vP}(\Psi_{\vP,E_{\vP}} ) \| \leq C_I$)  and the Cauchy-Schwartz inequality, we conclude that \eqref{eq:III.22} is bounded by  $C_I \| f_{\vec{Q}}^g(\vP)\|^{}_2$.

For  (\ref{eq:III.23}), we use again  Cauchy-Schwartz and the bound (for some constant $C$)
\begin{equation}
\|\vnabla_{\vP} f_{\vec{Q}}^g(\vP)\|^{}_2\leq C |\sup \vnabla_{\vP} f_{\vec{Q}}^g(\vP)| \|f_{\vec{Q}}^g(\vP)\|^{}_2 = \cO(|g|^{-\gamma}  \|f_{\vec{Q}}^g(\vP)\|^{}_2), 
\end{equation}
which can be checked from the construction of the functions $f_{\vec{Q}}^g$  (see  below Eq. (\ref{eq:II.8})).

Collecting the bounds on (\ref{eq:III.22}, \ref{eq:III.23}, \ref{eq:III.24}), we have proven the lemma.
\end{proof}
.

\subsection{Proof of Lemma \ref{lem-III.2} }  \label{appsec: lemma virial}

We now proceed with the proof of Lemma \ref{lem-III.2} in Section \ref{Section-V}.

\begin{blanklemma} [\textrm{\ref{lem-III.2}}] 
The identity
\begin{eqnarray}
0&= &(\Psi_{f_{\vec{Q}}^g}\,,\,d\Gamma(\chi_{n}^2(|\vk|)\xi^{g\,2}_{\hat{u}}(\hat{k})\,|\vk|)\Psi_{f_{\vec{Q}}^g})  \label{eq: limiting virial 1}  \\
 && -(\Psi_{f_{\vec{Q}}^g}\,,\,\vnabla E_{\vP}\cdot d\Gamma(\chi_{n}^2(|\vk|)\xi^{g\,2}_{\hat{u}}(\hat{k})\,\vk)\Psi_{f_{\vec{Q}}^g})   \label{eq: limiting virial 2}\\
& &   -g(\Psi_{f_{\vec{Q}}^g}\,,\,[a^*(id_{n}^{\hat{u}} \rho_{\vx})+\, a(id_{n}^{\hat{u}} \rho_{\vx})]\Psi_{f_{\vec{Q}}^g})\,  \label{eq: limiting virial 3}
\end{eqnarray}
holds true. As the one-particle state $\Psi_{f_{\vec{Q}}^g}$ belongs to the form domain of all operators in (\ref{eq: limiting virial 1}, \ref{eq: limiting virial 2}, \ref{eq: limiting virial 3}), this RHS is well-defined.
\end{blanklemma}

Since the dilation operator is unbounded, we must check that a regularized expression for the commutator  $\,i[H-E_{\vP}\,,\,D_{n}^{\hat{u}}]$ in Eq. (\ref{eq:III.9}) is well defined and that, upon the removal of the regularization, the expectation value of that commutator in the state $\Psi_{f_{\vec{Q}}^g}$  corresponds to the right-hand side above, i.e.\ (\ref{eq: limiting virial 1}, \ref{eq: limiting virial 2}, \ref{eq: limiting virial 3}). We show that,  provided $\beta$ is sufficiently large, the same strategy as implemented in \cite{FP}  justifies this identity.
Most of the arguments below, with the exception of the one in Section \ref{sec: the derivative term},  are standard in the literature. 

However, compared to the literature, our virial theorem has a little twist. This is due to the fact that we do not attempt to rule out any eigenvector, but merely an eigenvector with a certain regularity property. This is exploited in Lemma \ref{lem-III.1} and it is a crucial ingredient of the justification of the virial identity in Lemma \ref{lem-III.2}.

In Section \ref{sec: well defined rhs}, we prove that the expressions in  (\ref{eq: limiting virial 1}, \ref{eq: limiting virial 2}, \ref{eq: limiting virial 3}) are well-defined. 
In Section \ref{virial identity with regularization}, we start the proof of the equality in Lemma \ref{lem-III.2}.

\subsubsection{Well-definedness of the terms (\ref{eq: limiting virial 1}, \ref{eq: limiting virial 2}, \ref{eq: limiting virial 3}) } \label{sec: well defined rhs}

The operators
\begin{equation}
\qquad \,d\Gamma(\chi_{n}^2(|\vk|)\xi^{g\,2}_{\hat{u}}(\hat{k})\,|\vk|)  \quad \textrm{and} \quad  \vnabla_{\vP}E_{\vP}\cdot d\Gamma(\chi_{n}^2(|\vk|)\xi^{g\,2}_{\hat{u}}(\hat{k})\,\vk)
\end{equation}
are bounded by a (multiple of) $H^f$.  In fact, the operator $ \vnabla E_{\vP}$ is surely bounded if we restricted the total Hilbert space to the fibers $\vP \in I$.  This restriction can be done since the function $f_{\vec{Q}}^g$ has support in $I$.
Since 
\begin{equation} \label{eq: chain of domains}
\Psi_{f_{\vec{Q}}^g}\in \dom(H)\,\Rightarrow\,\Psi_{f_{\vec{Q}}^g}\in
\dom(H^f)
\end{equation} the expressions \eqref{eq: limiting virial 1} and \eqref{eq: limiting virial 2} are well-defined. 
Next, from the expression in \eqref{eq: splitting of form factor} and  the fact that $\rho \in C_1$, we have 
\begin{equation}  \label{eq: l infinity l two bound}
  \int d^3k \frac{1}{|\vk|} \sup_{\vx}  \left\str  \frac{1}{|\vx|+1}  ( d_n^{\hat{u}}\rho_{\vx})(\vk) \right\str^2 < \infty.
\end{equation} 
and hence, by a standard argument for bounding creation/annihilation operators, 
 \begin{equation}
\| \frac{1}{|\vx|+1} a(id_n^{\hat{u}}\rho_{\vx}) \frac{1}{(H^f+1)} \| < \infty.
\end{equation}
Since $\Psi_{f_{\vec{Q}}^g} \in \dom(\vx) \cap \dom(H^f)$ by  Lemma \ref{lem-III.1} and \eqref{eq: chain of domains}, it follows that 
 also the expression \eqref{eq: limiting virial 3} makes sense.

\subsubsection{Virial Identity with a regularized dilation operator} \label{virial identity with regularization}

We introduce the regularized gradient 
\begin{equation} \label{def: regularized gradient}
\vnabla^{\epsilon}_{\vk}:=\frac{\vnabla_{\vk}}{1-\epsilon\Delta_{\vk}}\,,
\end{equation}
where the parameter $\epsilon>0$  will be eventually removed. 
\noindent
Consequently, we also define $D_n^{\hat{u},\epsilon}:=d\Gamma(d_n^{\hat{u},\epsilon})$ where $d_n^{\hat{u},\epsilon}$ corresponds to $d_n^{\hat{u}}$ with $\vnabla_{\vk}$ replaced by $\vnabla^{\epsilon}_{\vk}$. 
Since, thanks to the regularization, $D_n^{\hat{u},\epsilon}$ is bounded w.r.t. to $H^f$, we deduce that  $\Psi_{f_{\vec{Q}}^g}\in \dom(D_n^{\hat{u},\epsilon})$ (cfr.\ \eqref{eq: chain of domains}).

We claim that
\begin{eqnarray}
&&  i((H-E_{\vP})\Psi_{f_{\vec{Q}}^g}\,,\,\,D_n^{\hat{u},\epsilon}\,\Psi_{f_{\vec{Q}}^g})-i(D_n^{\hat{u},\epsilon}\Psi_{f_{\vec{Q}}^g}\,,\,(H-E_{\vP})\Psi_{f_{\vec{Q}}^g})  \qquad  \qquad \qquad  \qquad \label{eq:II.25} \\
&&  \qquad  \qquad = \, \, (\Psi_{f_{\vec{Q}}^g}\,,\,d\Gamma(i[|\vk|,\,d^{\hat{u},\epsilon}_{n}])  \Psi_{f_{\vec{Q}}^g})  \label{eq:II.25andhalf}  \\
& &   \qquad  \qquad \quad-i(\Psi_{f_{\vec{Q}}^g}\, [E_{\vP}, D_n^{\hat{u},\epsilon}]\Psi_{f_{\vec{Q}}^g})\quad\quad\quad \quad\label{eq:II.26} \\
& &  \qquad  \qquad  \quad -g(\Psi_{f_{\vec{Q}}^g}\,,\,[a^*(id_n^{\hat{u},\epsilon}\rho_{\vx})+\, a(id_n^{\hat{u},\epsilon}\rho_{\vx})]\Psi_{f_{\vec{Q}}^g})\label{eq:II.27}\,
\end{eqnarray}
where  the LHS makes sense since $\Psi_{f_{\vec{Q}}^g}\in \dom(D_n^{\hat{u},\epsilon})$ and the RHS is obtained by \emph{formal} evaluation of the commutator  $[H-E_{\vP}, D_n^{\hat{u},\epsilon}]$.
All terms on the RHS are well-defined by similar (but easier) arguments as those in Section \ref{sec: well defined rhs} (for example, note that $[|\vk|,\,d^{\hat{u},\epsilon}_{n}]$ is a bounded operator). Nevertheless, the equality above requires a  justification. In the case at hand, a pedestrian way to provide such a justification is to introduce cutoffs in $\vx,\vk$ and $N$ (the number operator) such that all operators involved are bounded, calculate the commutator and finally remove the cutoffs.

Since $(H-E_{\vP})\Psi_{f_{\vec{Q}}^g}=0$ by assumption,  the expression  (\ref{eq:II.25}) vanishes.
Thus, it is enough to prove that the expressions (\ref{eq:II.25andhalf}, \ref{eq:II.26},  \ref{eq:II.27}) converge to (\ref{eq: limiting virial 1}, \ref{eq: limiting virial 2}, \ref{eq: limiting virial 3}), respectively, as $\epsilon$ tends to $0$. 
These three convergence statements will be established Sections \ref{sec: the energy term}, \ref{sec: the derivative term} and \ref{sec: the interaction term}, respectively.

\subsubsection{Some properties of the regularized dilation operator}\label{sec: properties of diffdilation}

In this preparatory section, we state some estimates on
\begin{equation}
e^{i\vz\cdot\vP}\,D_n^{\hat{u},\epsilon}e^{-i\vz\cdot\vP}-D_n^{\hat{u},\epsilon}
\end{equation}
that will be useful in taking the limit $\epsilon \rightarrow 0$. 
First, we remark that 
\begin{equation}
e^{i\vz\cdot\vP}\,D_n^{\hat{u},\epsilon}e^{-i\vz\cdot\vP}=d\Gamma(d_{n,\vz}^{\hat{u},\epsilon}), \qquad d_{n,\vz}^{\hat{u},\epsilon}:=e^{i\vz\cdot\vk}d_{n}^{\hat{u},\epsilon}  e^{-i\vz\cdot\vk}
\end{equation}
on the appropriate domain. Explicitly,
\begin{equation} 
d_{n,\vz}^{\hat{u},\epsilon}=\,\chi_{n}(|\vk|)\xi_{\hat{u}}^g(\hat{k})\frac{1}{2}\,(\vk\cdot\, \vec{F}_{\epsilon}(i\vnabla_{\vk}+\vz)+\vec{F}_{\epsilon}(i\vnabla_{\vk}+\vz) \cdot\,\vk)\xi_{\hat{u}}^g(\hat{k})\chi_{n}(|\vk|)\, \label{eq: expression for d ep}
\end{equation}
and $\vec{F}_{\epsilon}$ is the family of $\mathbb{R}^3 \mapsto \mathbb{R}^3$ functions given by (cfr.\ \eqref{def: regularized gradient})
\begin{equation}
\vec F_{\ep}(\vy) =  \frac{\vy}{1+\ep |\vy|^2}.
\end{equation}

We define the vector operator $\vdiffdilation_{n,\vz}^{\hat{u},\epsilon}$  such that it satisfies
\begin{equation} 
\vz \cdot \vdiffdilation_{n,\vz}^{\hat{u},\epsilon} = d_{n,\vz}^{\hat{u},\epsilon}-d_{n}^{\hat{u},\epsilon}\,. \label{def: diffdilation}
\end{equation}
Namely,
\begin{equation}
(\vdiffdilation_{n,\vz}^{\hat{u},\epsilon})_j :=  \,\chi_{n}(|\vk|)\xi_{\hat{u}}^g(\hat{k})\frac{1}{2}\,\sum_l\,\big(k_l\,  \int_{0}^{1}dt\, ( \vnabla F_{\ep,l})_j( i\vnabla_{\vk}+t\vec{z})\big)\xi_{\hat{u}}^g(\hat{k})\chi_{n}(|\vk|) + h.c.\,\,.
\end{equation}
where the subscripts $l$ and $j$ label vector components. 
To check that \eqref{def: diffdilation} holds, we substitute the line integral
\begin{equation}
 F_{\ep,l}(\vy+\vz)- F_{\ep,l}(\vy)=\vz\cdot    \int_{0}^{1}dt\,  \vnabla F_{\ep,l}( \vy+t\vec{z}), \qquad  l=1,2,3,
\end{equation} 
into the explicit expression for \eqref{eq: expression for d ep}, using the functional calculus.

We derive immediately the following properties
\begin{enumerate}
\item
The operator norms
\begin{equation}
  \|\,\vdiffdilation_{n,\vz}^{\hat{u},\epsilon}\,  \|, \qquad    \|\, \vk \chi_{n}^2(|\vk|)\xi_{\hat{u}}^{g\,{\,2}}(\hat{k}) \,  \|,   \label{boundedness of diffdilation}
\end{equation}
and hence also
\begin{equation}
 \|\,d\Gamma (\vdiffdilation_{n,\vz}^{\hat{u},\epsilon})\frac{1}{(H^f+1)}  
\,\|  \,    \qquad  \|d\Gamma( \vk\chi_{n}^2(|\vk|)\xi_{\hat{u}}^{g\,{\,2}}(\vk))\frac{1}{(H^f+1)}\|     \label{boundedness of second diffdilation}
 ,
\end{equation}
are bounded uniformly  in $\epsilon$ and in $\vz\in \RR^3$.  For the operators on the left (involving $\vdiffdilation_{n,\vz}^{\hat{u},\epsilon}$), this follows from the fact that $ \sup_{\vy,\ep} \|  \vnabla F_{\ep,j}(\vy)  \|  $ is bounded. For the operators on the right, this is a trivial consequence of the momentum cutoff functions. 
\item For each $\vz$,
\begin{eqnarray}
\vdiffdilation_{n,\vz}^{\hat{u},\epsilon}\quad &  \mathop{\longrightarrow}\limits^{\textrm{strongly}}_{\ep \rightarrow 0} & \quad \vk \chi_{n}^2(|\vk|)\xi_{\hat{u}}^{g\,{\,2}}(\hat{k}).  \label{eq: convergence of diffdilation} \\
d\Gamma (\vdiffdilation_{n,\vz}^{\hat{u},\epsilon})\frac{1}{(H^f+1)}  
\quad &  \mathop{\longrightarrow}\limits^{\textrm{strongly}}_{\ep \rightarrow 0} & \quad   d\Gamma( \vk\chi_{n}^2(|\vk|)\xi_{\hat{u}}^{g\,{\,2}}(\vk))\frac{1}{(H^f+1)}\quad\quad  \label{eq: convergence of second quantized diffdilation}
\end{eqnarray}
This convergence on $\dom (\vnabla_{\vk})$  and  $\dom (d \Gamma(\vnabla_{\vk})) \cap \cF_{fin}$  follows by $\vnabla F_{\ep,j}(\vy) \, \rightarrow \hat{y}_j$, as $\ep \rightarrow 0$, pointwise in $\vy$. Convergence on all vectors then follows by using the uniform boundedness  (\ref{boundedness of diffdilation}, \ref{boundedness of second diffdilation}) above.

\end{enumerate}

\subsubsection{The term $[H^f, D_n^{\hat{u}}]$}   \label{sec: the energy term}

In this section, we show  that (\ref{eq:II.25andhalf}) converges to \eqref{eq: limiting virial 1}, as $\epsilon \rightarrow 0$.

We derive
\begin{equation}
 d \Gamma (i[\,|\vk|\,,\,d_n^{\hat{u},\epsilon}] ) \frac{1}{H^f+1}  \qquad    \mathop{\longrightarrow}\limits_{\ep \rightarrow 0}^{\textrm{strongly}} \qquad  d\Gamma \left(  \,|\vk|\chi_{n}^2(|\vk|)\xi_{\hat{u}}^{g\,{\,2}}(\hat{k})   \right)  \frac{1}{H^f+1}  \label{eq: uniform boundedness energy}
\end{equation}
in exactly the same way as we did to arrive at \eqref{eq: convergence of second quantized diffdilation}. That is, we first establish (using properties of $F_\ep$) that 
$$ \sup_{\ep} \norm  i[\,|\vk|\,,\,d_n^{\hat{u},\epsilon}] \norm < \infty, $$ 
and that, on the dense domain $\dom(\vnabla_{\vk})$, the operator $ i[\,|\vk|\,,\,d_n^{\hat{u},\epsilon}]$  converges to $\,|\vk|\chi_{n}^2(|\vk|)\xi_{\hat{u}}^{g\,{\,2}}(\hat{k})$. 
Since  $\Psi_{f_{\vec{Q}}^g} \in \dom (H^f)$, we conclude that
\begin{equation}
d\Gamma\left(i[|\vk|,\,d^{\hat{u},\epsilon}_{n}] \right) \Psi_{f_{\vec{Q}}^g}  \qquad    \mathop{\longrightarrow}\limits_{\ep \rightarrow 0} \qquad  d\Gamma \left(  \,|\vk|\chi_{n}^2(|\vk|)\xi_{\hat{u}}^{g\,{\,2}}(\hat{k})   \right)   \Psi_{f_{\vec{Q}}^g} 
\end{equation}
We have proven that the difference between
 (\ref{eq:II.25andhalf}) and \eqref{eq: limiting virial 1} vanishes as $\epsilon \rightarrow 0$.

\subsubsection{The term $[E_{\vP}, D_n^{\hat{u}}]$}  \label{sec: the derivative term}

In this section, we show that \eqref{eq:II.26} converges to \eqref{eq: limiting virial 2}, as $\ep \rightarrow 0$.

We consider an extension of  the function $E_{\vP}$,  that is twice differentiable (see Section \ref{sec: properties}) and of compact support $\mathcal{K}$ (i.e., $\{\vP\,|\,|\vP|\in I\}\subset \mathcal{K}$). We use the same symbol, $E_{\vP}$,  for the function extended  to $\mathcal{K}$, and we write 
\begin{equation}
E_{\vP}= \int\,d^3z\,\hat{E}(\vz)e^{i\vz\cdot\vP}\,,
\end{equation}
where $\hat{E}(\vz)$ is the Fourier transform of $E_{\vP}$ (up to the prefactor $(2\pi)^{-3/2} $). Since $E_{\vP}$ is twice differentiable and of compact support,  $|\vz|^2\hat{E}(\vz)$  belongs to $L^2(\RR^3;d^3z)$ and, by Cauchy-Schwarz, $ \hat{E}(\vz)$ is in  $L^1(\RR^3;d^3z)$. Therefore, using the functional calculus, we can write,
\begin{eqnarray}
(\Psi_{f_{\vec{Q}}^g}\,,\, [E_{\vP}, \,D_n^{\hat{u},\epsilon} ]\,\Psi_{f_{\vec{Q}}^g}) = \int\,d^3z\,\hat{E}(\vz)(\Psi_{f_{\vec{Q}}^g}\,,\, [e^{i\vz\cdot\vP}, \,D_n^{\hat{u},\epsilon} ]\,\Psi_{f_{\vec{Q}}^g}) \label{eq: functional calculus for D} \,.
\end{eqnarray}

Then we observe that, on e.g.\ the domain $\dom (H^f)$,
\begin{equation} 
e^{i\vz\cdot\vP}\,D_n^{\hat{u},\epsilon}-D_n^{\hat{u},\epsilon}e^{i\vz\cdot\vP}=(e^{i\vz\cdot\vP}\,D_n^{\hat{u},\epsilon}e^{-i\vz\cdot\vP}
-D_n^{\hat{u},\epsilon})e^{i\vz\cdot\vP} = \vz \cdot d \Gamma (\vdiffdilation_{n,\vz}^{\hat{u},\epsilon}) e^{i\vz\cdot\vP}   \label{eq: translated D}
\end{equation} 
with the bounded operator $\vdiffdilation_{n,\vz}^{\hat{u},\epsilon}$ as defined in Section \ref{sec: properties of diffdilation}.
We are now ready to compare \eqref{eq:II.26} with \eqref{eq: limiting virial 2}:  
\begin{eqnarray}
&& i(\Psi_{f_{\vec{Q}}^g}\,,\, [E_{\vP}, \,D_n^{\hat{u},\epsilon} ]\,\Psi_{f_{\vec{Q}}^g}) -   (\Psi_{f_{\vec{Q}}^g}\,,\,\vnabla E_{\vP}\cdot d\Gamma(\chi_{n}^2(|\vk|)\xi^{g\,2}_{\hat{u}}(\hat{k})\,\vk)\Psi_{f_{\vec{Q}}^g}) \\
&=&  i \int\,d^3z\,\hat{E}(\vz)  (\Psi_{f_{\vec{Q}}^g}\,, \left(d\Gamma(\vdiffdilation_{n,\vz}^{\hat{u},\epsilon})  -\, d\Gamma( \vk\chi_{n}^2(|\vk|)\xi_{\hat{u}}^{g\,{\,2}}(\vk) \right) \cdot \vz \, e^{i \vz \cdot \vP}  \,\Psi_{f_{\vec{Q}}^g})\,\quad\quad\quad\quad   \label{eq:commutator E} \\
&=&-i\int\,d^3z\,\hat{E}(\vz)( \vx \Psi_{f_{\vec{Q}}^g}\,,\,[d\Gamma(\vdiffdilation_{n,\vz}^{\hat{u},\epsilon})- d\Gamma(\vk\chi_{n}^2(|\vk|)\xi_{\hat{u}}^{g\,{\,2}}(\vk)
)]e^{i\vz\cdot\vP}\,\Psi_{f_{\vec{Q}}^g})   \label{eq: commutator x D} \\
&&+i \int\,d^3z\,\hat{E}(\vz)(\Psi_{f_{\vec{Q}}^g}\,,\,[d\Gamma(\vdiffdilation_{n,\vz}^{\hat{u},\epsilon})- d\Gamma(\vk\chi_{n}^2(|\vk|)\xi_{\hat{u}}^{g\,{\,2}}(\vk)
)]e^{i\vz\cdot\vP}\, \vx \Psi_{f_{\vec{Q}}^g}) \label{eq: commutator D x} 
\end{eqnarray}
The first equality follows by   (\ref{eq: functional calculus for D}, \ref{eq: translated D}, \ref{def: diffdilation})  and the fact that the Fourier transform sends differentiation into multiplication.
To obtain the second equality, we used the canonical commutation relation 
\begin{equation}\label{Weyl}
\vz  e^{i\vz\cdot\vP}=[e^{i\vz\cdot\vP}, \vx]
\end{equation}
which holds e.g.\ on $\dom (\vx) \cap \dom(H^f)$.  

Since  $\hat{E}(\vz)\in L^1(\RR^3;d^3z)$, we can estimate \eqref{eq: commutator x D}  
\begin{equation}
 |\eqref{eq: commutator x D}|
\leq \int\,d^3z\,|\hat{E}(\vz)|\, \left|(\vx \Psi_{f_{\vec{Q}}^g}\,,\,[d\Gamma(\vdiffdilation_{n,\vz}^{\hat{u},\epsilon})- d\Gamma(\vk\chi_{n}^2(|\vk|)\xi_{\hat{u}}^{g\,{\,2}}(\vk)\nonumber
)]e^{i\vz\cdot\vP}\,\Psi_{f_{\vec{Q}}^g})\right |\,
\end{equation}
For each $\vec{z}$, the second factor vanishes as $\ep \rightarrow 0$ by  \eqref{eq: convergence of second quantized diffdilation} and the fact that $\Psi_{f_{\vec{Q}}^g} \in \dom (\vx) \cap \dom(H^f)$.
Hence we  conclude that (\ref{eq: commutator x D} ) tends to zero as $\epsilon$ tends to zero, by dominated convergence.
Obviously, \eqref{eq: commutator D x}  can be treated in exactly the same way and hence we have proven that \eqref{eq:commutator E} vanishes as $\ep \rightarrow 0$. 

Hence, we have shown that the difference between \eqref{eq:II.26} and  \eqref{eq: limiting virial 2} vanishes, as $\ep \rightarrow 0$.

%

\subsubsection{The term $[g\phi(\rho_{\vx}), D_n^{\hat{u}}]$} \label{sec: the interaction term}

In this section, we prove that \eqref{eq:II.27} converges to \eqref{eq: limiting virial 3} as $\epsilon \downarrow 0$. 

First, we note that 
\begin{equation}  \label{eq: l infinity l two bound ep}
\sup_{\ep} \int d^3k \frac{1}{|\vk|}\sup_{\vx}   \left\str  \frac{1}{|\vx|+1}  ( d_n^{\hat{u}, \ep}\rho_{\vx})(\vk) \right\str^2 < \infty,
\end{equation} 
This follows in the same way as \eqref{eq: l infinity l two bound} , established in Section \ref{sec: well defined rhs}.
Together, \eqref{eq: l infinity l two bound ep} and \eqref{eq: l infinity l two bound} imply that the operator norms of 
\begin{equation}
  R_{n,\ep}^{\hat u} :=  \frac{1}{H^f+1}  \,a^*(i(d_n^{\hat{u},\epsilon}-d_n^{\hat{u}}) \rho_{\vx}) \frac{1}{|\vx|+1}
 \end{equation}
 \begin{equation}
  (R_{n,\ep}^{\hat u})^* :=  \frac{1}{|\vx|+1}\, a(i(d_n^{\hat{u},\epsilon}-d_n^{\hat{u}})\rho_{\vx}) \frac{1}{H^f+1}
 \end{equation}
are  uniformly bounded in $\epsilon$. We can now take advantage of the fact that $\Psi_{f_{\vec{Q}}^g} \in \dom (\vx) \cap \dom(H^f)$ to write 
\begin{eqnarray}
&& (\Psi_{f_{\vec{Q}}^g}\,,\,[a^*(i(d_n^{\hat{u},\epsilon}-d_n^{\hat{u}}) \rho_{\vx})+\, a(i(d_n^{\hat{u},\epsilon}-d_n^{\hat{u}})\rho_{\vx})]\Psi_{f_{\vec{Q}}^g}) \label{eq: difference of dilations} \\
&=&  ((H^f +1) \Psi_{f_{\vec{Q}}^g}\,,\, R_{n,\ep}^{\hat u}\chi^{}_{K_{\delta}} (|\vx|+1) \Psi_{f_{\vec{Q}}^g})   \label{eq: the part in kdelta} \\
&&+  (   (H^f +1)\Psi_{f_{\vec{Q}}^g}\,, R_{n,\ep}^{\hat u} (|\vx|+1)(1- \chi^{}_{K_{\delta}})   \, \Psi_{f_{\vec{Q}}^g})
 \label{eq: the part out of kdelta} \\
& &+ ((|\vx|+1)\Psi_{f_{\vec{Q}}^g}\,,\,  \chi^{}_{K_{\delta}}(R_{n,\ep}^{\hat u})^*(H^f +1)  \Psi_{f_{\vec{Q}}^g})   \label{eq: the part in kdelta-2} \\
& &+ ((1- \chi^{}_{K_{\delta}})(|\vx|+1)\Psi_{f_{\vec{Q}}^g}\,,\, (R_{n,\ep}^{\hat u})^*(H^f +1)  \Psi_{f_{\vec{Q}}^g})   \label{eq: the part out of kdelta-2} 
\end{eqnarray}
where  $\chi^{}_{K_{\delta}}= \chi^{}_{K_{\delta}}(\vx)$ is the characteristic function of a compact set $K_\delta \subset \mathbb{R}^3$, chosen such that the $\left\str\textrm{\eqref{eq: the part out of kdelta}}\right\str$,  $\left\str\textrm{\eqref{eq: the part out of kdelta-2}}\right\str$ are smaller than $\delta$. This can be done by the uniform bound on  $ \norm R_{n,\ep}^{\hat u} \norm$,  and the fact that $\norm (\chi^{}_{K_{\delta}}-1)(|\vx|+1) \Psi_{f_{\vec{Q}}^g} \norm $ can be made arbitrarily small by choosing $K_{\delta}$ big enough. 
Moreover, for any compact $K$, 
\begin{equation}
\lim_{\epsilon\to0}\int d^3k\,\frac{1}{|\vk|}\,[\sup_{\vx\in K}|((id_n^{\hat{u},\epsilon}\,-\,id_n^{\hat{u}})\rho_{\vx})(\vk)|]^2=0\,,
\end{equation}
This implies that  $ \norm  \chi^{}_{K}  R_{n,\ep}^{\hat u} \norm $, $ \norm  \chi^{}_{K}  (R_{n,\ep}^{\hat u})^*\norm $ and hence \eqref{eq: the part in kdelta}, \eqref{eq: the part in kdelta-2} vanish, as $\ep \rightarrow 0$.
Together, the bounds on \eqref{eq: the part in kdelta}, \eqref{eq: the part in kdelta-2}  and on \eqref{eq: the part out of kdelta}, \eqref{eq: the part out of kdelta-2} prove that \eqref{eq: difference of dilations} vanishes in the limit $\ep \rightarrow 0$. 
Hence, the difference of \eqref{eq:II.27}  and \eqref{eq: limiting virial 3} vanishes as $\epsilon \downarrow 0$.

\subsection{Proof of the fiber virial identity in \eqref{final-virial}}\label{sec: proof of virial fiber}

The justification of the  virial identity in \eqref{final-virial}  is largely analogous to that of the virial identity in Lemma \ref{lem-III.2}.
To avoid repetitive arguments, we just sketch the main strategy of the proof.

First one introduces a regularized dilation operator $d_{\frac{1}{\kappa},\kappa}^{\epsilon}$ and the corresponding second quantized operator $D^{b,\epsilon}_{\frac{1}{\kappa},\kappa}:=d\Gamma^b(d_{\frac{1}{\kappa},\kappa}^{\epsilon})$.  The operator  $d_{\frac{1}{\kappa},\kappa}^{\epsilon}$ is obtained from  $d_{\frac{1}{\kappa},\kappa}$ (see Eq. (\ref{eq:III.7}))  by replacing  the gradient, $\vnabla_{\vk}$,   with  
\begin{equation}
\vnabla^{\epsilon}_{\vk}:=\frac{\vnabla_{\vk}}{1-\epsilon\Delta_{\vk}}\,,\quad \epsilon>0.
\end{equation} Then one exploits the following properties:
\begin{itemize}
\item[i)]
On the dense subspace $\dom(\vnabla_{\vk}) \in \mathfrak{h}$,
\begin{equation}
i[\,|\vk|\,,\,d_{\frac{1}{\kappa},\kappa}^{\epsilon}]\rightarrow\,|\vk|\chi_{[\frac{1}{\kappa},\kappa]}^2(|\vk|)\,,
\end{equation}
\begin{equation}
i[\,\vk\,,\,d_{\frac{1}{\kappa},\kappa}^{\epsilon}]\rightarrow\,\vk\chi_{[\frac{1}{\kappa},\kappa]}^2(|\vk|)
\end{equation}
as $\epsilon \rightarrow 0$. (Strong convergence on the whole of $ \mathfrak{h}$ follows than from ii) below).
\item[ii)]
The operator norms
\begin{eqnarray}
&& \|\,[\,|\vk|\,,\,d_{\frac{1}{\kappa},\kappa}^{\epsilon}]\,\| ,\quad \|\,[\,\vk\,,\,d_{\frac{1}{\kappa},\kappa}^{\epsilon}]\,\|     \\[1mm]
&& \|d\Gamma^b(i[|\vk|,\,d^{\epsilon}_{\frac{1}{\kappa},\kappa}])\,\frac{1}{(1+H^f)}\|,\quad   \|d\Gamma^b(i[\vk,\,d^{\epsilon}_{\frac{1}{\kappa},\kappa}])\,\frac{1}{(1+H^f)}\|\quad\quad
\end{eqnarray}
are bounded uniformly  in $\epsilon$. 
\item[iii)]  
\begin{equation}
\lim_{\epsilon\to0}\int d^3k\,\frac{1}{|\vk|}|(id_{\frac{1}{\kappa},\kappa}^{\epsilon}\,-\,id_{\frac{1}{\kappa}, \kappa})\rho(\vk)|^2=0\,.
\end{equation}

\item[iv)] 
the operator norm 
\begin{equation}
\|b(id_{\frac{1}{\kappa}, \kappa}^{\epsilon}\rho)\frac{1}{(H^f+1)^{1/2}}\|
\end{equation}
is uniformly bounded in $\epsilon$.

\item[v)] 
\begin{equation}
\left\norm b(id_{(\frac{1}{\kappa}, \kappa}^{\epsilon}\,-\,id_{\frac{1}{\kappa},\kappa})\rho(\vk))\, \frac{1}{(1+H^f)^{1/2}} \right\norm^2 
\leq \int d^3k\,\frac{1}{|\vk|}|(id_{\frac{1}{\kappa},\kappa}^{\epsilon}\,-\,id_{\frac{1}{\kappa},\kappa})\rho(\vk))|^2  \quad
\end{equation}

\end{itemize}


\begin{thebibliography}{12}

\bibitem{minlos1}
N. Angelescu, R.A. Minlos, and V.A. Zagrebnov.
\newblock Lower spectral branches of a particle coupled to a Bose field.
\newblock  {\em Rev.~Math.~Phys}, {\bf 17}  (9), 1--32 (2005).

\bibitem{minlos2}
N. Angelescu, R.A. Minlos, and V.A. Zagrebnov.
\newblock Lower Spectral Branches of a Spin-Boson Model.
\newblock  {\em J.~ Math.~ Phys.}, {\bf 49} 102105  (2008).

\bibitem{BCFS2}
V.~Bach, T. Chen, J.~Fr{\"{o}}hlich, and I.M. Sigal.
\newblock The renormalized electron mass in Non-Relativistic Quantum Electrodynamics.
\newblock   {\em J. Funct. Anal.}, {\bf 243}  (2) 426--535  (2007).


\bibitem{C}
T.~Chen. Infrared renormalization in non-relativistic QED and scaling
criticality.  {\em J. Funct. Anal.}, {\bf 254}  (10) 2555--2647  (2007).

\bibitem{CF}
T.~Chen and J.~Fr\"ohlich.. Coherent infrared representations in nonrelativistic
QED.\emph{Spectral Theory and Mathematical Physics: A Festschrift in Honor of
  Barry Simon's 60th Birthday Proc. Symp. Pure Math.} AMS, 2007.


\bibitem{CFP1}
T.~Chen, J.~Fr{\"{o}}hlich and A.~Pizzo.
\newblock Infraparticle Scattering States in QED: II. Mass Shell properties. 
\newblock  \emph{J.~ Math.~ Phys.},  {\bf 50} 012103 (2009)

\bibitem{CFP2}
T.~Chen, J.~Fr{\"{o}}hlich and A.~Pizzo.
\newblock Infraparticle Scattering States in QED: I. The Bloch-Nordsieck Paradigm. 
\newblock   {\em Comm. Math. Phys.},  {\bf 294} (3), 761-825   (2010)


\bibitem{erdos}
L. Erd{\"{o}}s. 
\newblock  Linear Boltzmann Equation as the Long Time Dynamics of an Electron Weakly Coupled to a Phonon Field
\newblock {\em J.~Stat.~Phys}, 107 (5-6): 1043-1127, 2002

\bibitem{cyconfroesekirschsimon}
\newblock H.~L.~Cycon, R.~G.~Froese, W.~Kirsch and B.~Simon.  Schršdinger Operators, with Applications to Quantum Mechanics and Global Geometry.
\newblock Berlin, Springer-Verlag, 1987

\bibitem{F1}
J.~Fr\"ohlich. On the infrared problem in a model of scalar electrons and
massless, scalar bosons. \newblock {\em
  Inst.~Henri~Poincare,~Section~Physique~Th\'eorique}, 19 (1):1--103, 1973.

\bibitem{F2}
J.~Fr\"ohlich. Existence of dressed one electron states in a class of
persistent models. \newblock {\em
  Fortschritte der Physik}, {\bf{22}}, 159--198 (1974).
  
\bibitem{FP}
J.~Fr{\"{o}}hlich and A.~Pizzo.
\newblock On the Absence of Excited Eigenstates in QED. 
\newblock {\em Comm. Math. Phys.},   {\bf{286}} (3), 803--836  (2009).


\bibitem{FP2}
J.~Fr{\"{o}}hlich and A.~Pizzo.
\newblock The renormalized electron mass in non-relativistic  QED. 
\newblock  {\em Comm. Math. Phys.}   DOI 10.1007/s00220-009-0960-8

\bibitem{froehlichgriesemersigal}
J.~Fr{\"{o}}hlich, M.Griesemer, and I.M. Sigal.
\newblock Mourre estimate and spectral theory for the standard model of
non-relativistic QED.
\newblock {\em Mp-arc} 06-316


\bibitem{georgescugerard}
 \newblock  V.~Georgescu and C.~Gerard.    On the Virial Theorem in Quantum Mechanics
 \newblock {\em  Comm. Math. Phys. },  {\bf{208}} (2), 275--281 (1999).
 
 \bibitem{hasler}
D.~Hasler and I.~Herbst.
\newblock Absence of Ground States for a Class of Translation Invariant Models of Non-relativistic QED. 
\newblock {\em Comm. Math. Phys.}  {\bf{279}} (3), 769--787  (2008).


\bibitem{moeller}
 \newblock  J.~Schach-Moeller.  The Translation Invariant Nelson Model: I. The Bottom of the Spectrum
 \newblock {\em Ann.~H.~Poincar{\'e}}, , {\bf 6} (6), 1091--1135

\bibitem{P1}
\newblock A.~Pizzo. One Particle (improper) States in Nelson's Massless Model.
\newblock {\em Ann.~H.~Poincar{\'e}}, {\bf 4} (3), 439--486
(2003). 

\bibitem{P2}
\newblock A.~Pizzo. Scattering of an Infraparticle: The One Particle Sector in Nelson's
Massless Model.
\newblock {\em Ann.~H.~Poincar{\'e}}, {\bf 6}, 553--606 (2005).

\bibitem{spohn}
\newblock H.~Spohn. The polaron at large total momentum. 
\newblock {\em  J. Phys. A}, {\bf 21}, 1199--1212 (1988).

\end{thebibliography}
\end{document}